\documentclass[3p,11pt]{elsarticle}

\newcommand{\revisedA}[1]{#1\xspace}
\newcommand{\markRevisedA}{}

\usepackage{latexsym}
\usepackage{xspace}
\usepackage{amssymb}
\usepackage{amsmath}
\usepackage{stmaryrd}
\usepackage{hyperref}
\usepackage{url}
\usepackage{graphicx}
\input{xy}
\xyoption{all}
\usepackage{calc}
\usepackage{tikz}
\usepackage{flushend}
\usepackage{color}
\usepackage{framed}
\usepackage[ruled,vlined,linesnumbered]{algorithm2e}

\newtheorem{theorem}{Theorem}[section]
\newtheorem{lemma}[theorem]{Lemma}

\newtheorem{Definition}[theorem]{Definition}
\newtheorem{Example}[theorem]{Example}
\newtheorem{Remark}[theorem]{Remark}

\newenvironment{definition}{\begin{Definition}\begin{em}}{\end{em}\end{Definition}}
\newenvironment{example}{\begin{Example}\begin{em}}{\end{em}\end{Example}}
\newenvironment{remark}{\begin{Remark}\begin{em}}{\end{em}\end{Remark}}
\newproof{proof}{Proof}

\def\eqref#1{(\ref{#1})}

\def\tuple#1{\langle#1\rangle}

\newcommand{\E}{\exists}

\newcommand{\fALC}{\mbox{$f\!\mathcal{ALC}$}\xspace}

\newcommand{\myend}{\mbox{}\hfill{\footnotesize$\Box$}}

\newcommand{\comment}[1]{}

\newcommand{\fand}{\varotimes}

\newcommand{\fto}{\Rightarrow}
\newcommand{\fequiv}{\Leftrightarrow}

\newcommand{\SV}{\Sigma_V}
\newcommand{\SE}{\Sigma_E}
\newcommand{\Vmin}{V_{min}}
\newcommand{\Vmax}{V_{max}}

\newcommand{\mN}{\mathcal{N}}

\newcommand{\Pred}{\mathit{Pred}}
\newcommand{\Succ}{\mathit{Succ}}
\newcommand{\done}{\mathit{done}}
\newcommand{\counter}{\mathit{count}}
\newcommand{\False}{\mathit{false}}
\newcommand{\True}{\mathit{true}}
\newcommand{\obj}{\mathit{obj}}
\newcommand{\smallestKey}{\mathit{smallestKey}}

\SetKwInput{Input}{Input}
\SetKwInput{Output}{Output}
\SetKwInput{Purpose}{Purpose}
\SetKw{Break}{break}
\SetKw{Continue}{continue}
\SetKw{Return}{return}
\SetKw{New}{new\,}

\journal{arXiv}

\begin{document}
\sloppy
	
\begin{frontmatter}
		
\title{Computing markings for fuzzy minimax nets over the G\"odel structure}

\author{Linh Anh Nguyen}
\ead{nguyen@mimuw.edu.pl}
\ead{nalinh@ntt.edu.vn}				
\address{Institute of Informatics, University of Warsaw, Banacha 2, 02-097 Warsaw, Poland}
\address{Faculty of Information Technology, Nguyen Tat Thanh University, Ho Chi Minh City, Vietnam}

\begin{abstract}
Fuzzy minimax nets were recently introduced as a tool for computing the greatest fuzzy bisimulation and simulation between two finite fuzzy graph-based structures. In this work, we provide an efficient algorithm \revisedA{with time complexity $O(m + n + l\log{l})$} for computing the greatest correct marking of a finite fuzzy minimax net over the G\"odel structure\revisedA{, where $n$, $m$, and $l$ denote the numbers of nodes, positive edges, and distinct fuzzy values used in the net, respectively}. Building on this result, we derive the first algorithm with time complexity $O((m+n)n)$ for computing the greatest fuzzy directed simulation between two finite fuzzy graphs over the G\"odel structure, where $n$ and $m$ denote the total numbers of vertices and positive edges, respectively, in the input graphs.
\end{abstract}


\end{frontmatter}


\section{Introduction}
\label{section:intro}

Fuzzy minimax nets were introduced in~\cite{DBLP:journals/tfs/NguyenMS23} as a tool for computing the greatest fuzzy bisimulation and simulation between two finite fuzzy labeled graphs (FLGs). FLGs provide a common representation for fuzzy labeled transition systems, fuzzy automata, fuzzy Kripke models, fuzzy social networks, and fuzzy interpretations in description logic. A node in a fuzzy minimax net is either a min-node or a max-node. Each min-node is associated with an upper marking limit, which is a value in the unit interval $[0,1]$. Each edge connects a max-node to a min-node or vice versa and is assigned a fuzzy weight from $[0,1]$. The problem of computing the greatest fuzzy bisimulation (respectively, simulation) between two finite FLGs is reduced to computing the greatest correct marking of the fuzzy minimax net that corresponds to the given FLGs with respect to bisimulation (respectively, simulation). Using fuzzy minimax nets, the work~\cite{DBLP:journals/tfs/NguyenMS23} presented the first algorithms for computing the greatest fuzzy bisimulation and simulation between two finite FLGs over the product structure. It also introduced the first algorithms whose complexity order is independent of the fuzzy values occurring in the inputs for those tasks in the case of using the \L{}ukasiewicz structure.

The problem of computing the greatest correct marking of a fuzzy minimax net over the G\"odel structure was not investigated in~\cite{DBLP:journals/tfs/NguyenMS23}, because efficient algorithms for computing the greatest fuzzy bisimulation and simulation between two finite FLGs over the G\"odel structure were already available~\cite{TFS2020,DBLP:journals/isci/Nguyen23}. Nevertheless, studying this problem is worthwhile for the following reasons.
\begin{itemize}
\item First, fuzzy minimax nets are a recently introduced mathematical structure, and computing their greatest correct marking is one of their fundamental algorithmic problems. Consequently, this problem deserves investigation in its own right, independently of its applications to computing fuzzy behavioral relations.
\item Second, as shown in~\cite{FDSML}, fuzzy minimax nets can also be exploited to compute the greatest fuzzy directed simulation between two finite fuzzy Kripke models, a problem that had not previously been investigated.
\end{itemize}

\revisedA{Directed simulations between Kripke models characterize the preservation of positive modal formulas. They were originally introduced by Kurtonina and de~Rijke~\cite{KurtoninaR97}. In~\cite{FDSML}, this notion is extended to fuzzy directed simulations between fuzzy Kripke models over linear and complete residuated lattices, its logical characterizations are investigated, and a method is presented for computing the greatest fuzzy directed simulation between two finite fuzzy Kripke models, based on a reduction to the problem of computing the greatest correct marking of a finite fuzzy minimax net.}

In this work, we provide an efficient algorithm \revisedA{with time complexity $O(m + n + l\log{l})$} for computing the greatest correct marking of a finite fuzzy minimax net over the G\"odel structure\revisedA{, where $n$, $m$, and $l$ denote the numbers of nodes, positive edges, and distinct fuzzy values used in the net, respectively}. Building on this result, we \revisedA{specialize the general method of~\cite{FDSML} to the G\"odel structure and} derive an algorithm \revisedA{with time complexity $O((m+n)n)$} for computing the greatest fuzzy directed simulation between two finite fuzzy graphs over the G\"odel structure\revisedA{, where} $n$ and $m$ denote the total numbers of vertices and positive edges, respectively, in the input graphs. This complexity bound is comparable to that of the algorithms for computing the greatest simulation between two finite crisp graphs~\cite{BloomP95,HenzingerHK95}.

The remainder of the paper is organized as follows. Section~\ref{section: prel} introduces the necessary preliminaries. Section~\ref{section: alg1} presents the algorithm for computing the greatest correct marking of a finite fuzzy minimax net over the G\"odel structure, together with its correctness proof and complexity analysis. Section~\ref{section: alg2} presents the algorithm for computing the greatest fuzzy directed simulation between two finite FLGs over the G\"odel structure, together with its correctness proof and complexity analysis. Section~\ref{section: impl and performance} describes an implementation of the algorithms and reports the results of the performance evaluation. Section~\ref{section: related work} discusses related work. Finally, Section~\ref{section: conc} concludes the paper.


\section{Preliminaries}
\label{section: prel}

In this section, we recall basic definitions concerning fuzzy sets, fuzzy relations, and fuzzy minimax nets. The presentation is adapted from~\cite{DBLP:journals/tfs/NguyenMS23}.

\subsection{Fuzzy sets and operators}
\label{sec: prel-1}

We use the fuzzy operators $\land, \lor, \fand, \fto: [0,1] \times [0,1] \to [0,1]$ defined as follows: 
\[
a \land b = \min(a,b), \qquad
a \lor b = \max(a,b), \qquad
a \fand b = a \land b, \qquad
(a \fto b) =
\begin{cases}
b & \text{if } a > b,\\
1 & \text{otherwise}.
\end{cases}
\]
They form the so-called G\"odel structure (of fuzzy truth values). These operations are called the \emph{meet}, \emph{join}, \emph{t-norm}, and \emph{residuum}, respectively.

A \emph{fuzzy subset} of a set $X$ is a function $f: X \to [0,1]$, which is also called a {\em fuzzy set}. 
It is \emph{empty} if $f(x) = 0$ for all $x \in X$. 
The {\em support} of $f$ is the set $\{x \in X \mid f(x) > 0\}$. 
For $x_1,\ldots,x_n \in X$ and $a_1,\ldots,a_n \in [0,1]$, by $\{x_1\!:\!a_1, \ldots, x_n\!:\!a_n\}$ we denote the fuzzy subset $f$ of $X$ such that $f(x_i) = a_i$ for each $1 \le i \le n$, and $f(x) = 0$ for all $x \in X \setminus \{x_1,\ldots,x_n\}$.

Given fuzzy subsets $f$ and $g$ of $X$, we say that $f$ is less than or equal to $g$, and write $f \leq g$, if $f(x) \leq g(x)$ for all $x \in X$. 
For a family $F$ of fuzzy subsets of $X$, by $\bigvee\!F$ we denote the fuzzy subset of $X$ defined by
$ (\bigvee\!F)(x) = \bigvee_{f \in F} f(x)$. 
We write $f \lor g$ to denote $\bigvee\{f,g\}$. 

A \emph{fuzzy relation} between $X$ and $Y$ is a fuzzy subset of $X \times Y$. 
Given a fuzzy relation \mbox{$\varphi: X \times Y \to [0,1]$}, the {\em converse} of $\varphi$ is the fuzzy relation \mbox{$\varphi^{-1} : Y \times X \to [0,1]$} defined by $\varphi^{-1}(y,x) = \varphi(x,y)$. 
Given fuzzy relations \mbox{$\varphi: X \times Y \to [0,1]$} and \mbox{$\psi: Y \times Z \to [0,1]$}, we define the {\em composition} $\varphi \circ \psi$ as the fuzzy relation $(\varphi \circ \psi): X \times Z \to [0,1]$ given by
\begin{equation}
(\varphi \circ \psi)(x,z) = \bigvee_{y \in Y} \left(\varphi(x,y) \fand \psi(y,z)\right).
\end{equation}


\subsection{Fuzzy minimax nets}
\label{sec: minimax}

A {\em (fuzzy) minimax net} is a structure $\mN = \tuple{\Vmin, \Vmax, E, L}$, where $\Vmin$ and $\Vmax$ are non-empty disjoint sets of {\em nodes}, consisting of {\em min-nodes} and {\em max-nodes}, respectively, $E: (\Vmin \times \Vmax) \cup (\Vmax \times \Vmin) \to [0,1]$ is a fuzzy set of {\em edges}, and $L: \Vmin \to [0,1]$ is the {\em marking limit for min-nodes}. We assume $E(x,y) = 0$ for $\tuple{x,y} \in (\Vmin \times \Vmin) \cup (\Vmax \times \Vmax)$. 
The net is {\em coimage-finite} if, for every $y \in \Vmin \cup \Vmax$, the set $\{x \in \Vmin \cup \Vmax \mid E(x,y) > 0\}$ is finite. 
A {\em positive} edge of $\mN$ is a pair $\tuple{x,y} \in (\Vmin \times \Vmax) \cup (\Vmax \times \Vmin)$ with $E(x,y) > 0$. 
By $|E|$ we denote the {\em size} of $E$, which is defined to be the number of positive edges of~$\mN$. 

An example of a minimax net is given in Figure~\ref{fig: HJRKA}. 

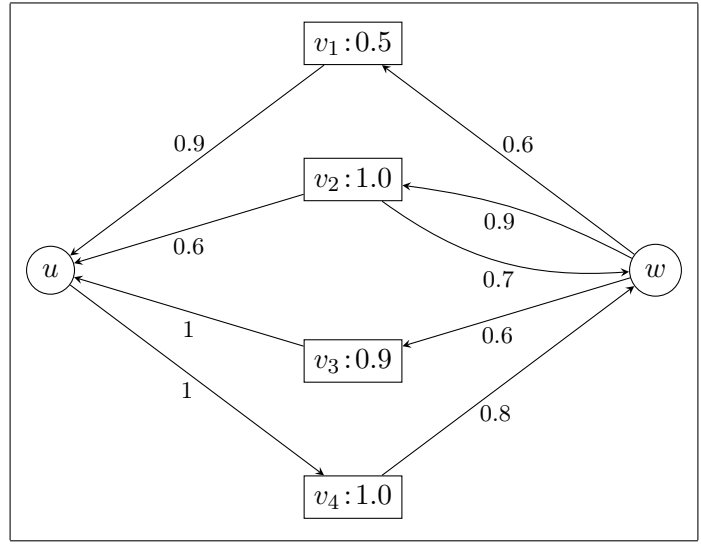
\begin{figure}[t]
\begin{center}
\begin{tabular}{|c|}
\hline
\begin{tikzpicture}[->,>=stealth,auto,black]
\node (top) {};
\node (v1) [draw, node distance=0.4cm, below of=top] {$v_1\!:\!0.5$};
\node (v2) [draw, node distance=1.8cm, below of=v1] {$v_2\!:\!1.0$};
\node (vm) [node distance=1.2cm, below of=v2] {};
\node (v3) [draw, node distance=1.2cm, below of=vm] {$v_3\!:\!0.9$};
\node (v4) [draw, node distance=1.8cm, below of=v3] {$v_4\!:\!1.0$};
\node (bot) [node distance=0.3cm, below of=v4] {};
\node (u)  [circle, draw, node distance=4.0cm, left of=vm] {$u$};
\node (w)  [circle, draw, node distance=4.0cm, right of=vm] {$w$};

\draw (v1) to node [above, pos=0.53, yshift=2]{\footnotesize{0.9}} (u);	
\draw (v2) to node [below]{\footnotesize{0.6}} (u);	
\draw (v3) to node [below]{\footnotesize{1}} (u);
\draw (u) to node [below, pos=0.46]{\footnotesize{1}} (v4);	

\draw (w) to node [above, pos=0.46, yshift=2]{\footnotesize{0.6}} (v1);	
\draw (w) edge[bend right=10] node [below, pos=0.60]{\footnotesize{0.9}} (v2);	
\draw (v2) edge[bend right=20] node[below, pos=0.5, xshift=0]{\footnotesize{0.7}} (w);
\draw (w) to node [below, pos=0.583]{\footnotesize{0.6}} (v3);
\draw (v4) to node [below, pos=0.45, yshift=-2]{\footnotesize{0.8}} (w);	
\end{tikzpicture}
\\
\hline
\end{tabular}
\caption{An illustration of a minimax net $\mN = \tuple{\Vmin, \Vmax, E, L}$, where $\Vmin = \{v_1,v_2,v_3,v_4\}$ (min-nodes are shown as rectangles) and $\Vmax = \{u,w\}$ (max-nodes are shown as circles). For each min-node $x$, the number displayed inside its rectangle (distinct from its label) denotes the marking limit $L(x)$. A positive edge $\tuple{x,y}$ is depicted as an arrow from $x$ to $y$ labeled by $E(x,y)$.\label{fig: HJRKA}}
\end{center}
\end{figure}

Given a minimax net $\mN = \tuple{\Vmin, \Vmax, E, L}$, a fuzzy set $M: \Vmin \cup \Vmax \to [0,1]$ is called a {\em marking} of $\mN$. It is a {\em correct marking} of $\mN$ if, for every $x \in \Vmin$ and $y \in \Vmax$, 
\begin{eqnarray}
	M(x) & \leq & L(x) \label{eq: JHDSD 1} \\
	M(x) & \leq & \bigwedge_{z \in \Vmax}\!\! (E(z,x) \fto M(z)) \label{eq: JHDSD 2} \\
	M(y) & \leq & \bigvee_{z \in \Vmin}\!\!(E(z,y) \fand M(z)). \label{eq: JHDSD 3} 
\end{eqnarray}
A marking $M$ of $\mN$ is {\em stable} if, for every $x \in \Vmin$ and $y \in \Vmax$, 
\begin{eqnarray}
	M(x) & = & L(x) \land \bigwedge_{z \in \Vmax}\!\!(E(z,x) \fto M(z)) \label{eq: JHDSD 4} \\
	M(y) & = & \bigvee_{z \in \Vmin}\!\!(E(z,y) \fand M(z)). \label{eq: JHDSD 5} 
\end{eqnarray}

Note that, by definition, stable markings are correct. 
It has been proved in~\cite{DBLP:journals/tfs/NguyenMS23} that every minimax net has the greatest correct marking and that the greatest correct marking of a minimax net is stable.

\begin{example}
Consider the minimax net $\mN$ specified in Figure~\ref{fig: HJRKA}. It is straightforward to verify that
\[ M = \{u\!:\!0.9, v_1\!:\!0.5, v_2\!:\!0.8, v_3\!:\!0.9, v_4\!:\!0.9, w\!:\!0.8\} \]
is a stable marking of $\mN$. By Example~\ref{example: HJRKA} and Theorem~\ref{theorem: JHDLZ}, stated in the next section, $M$ is the greatest correct marking of $\mN$.
\myend
\end{example}

\section{Computing the greatest correct marking of a minimax net}
\label{section: alg1}

In this section, we design an efficient algorithm for computing the greatest correct marking $M$ of a finite minimax net $\mN = \tuple{\Vmin, \Vmax, E, L}$ over the G\"odel structure. We begin by outlining the underlying idea, then present the algorithm formally and prove its correctness and complexity. We also provide an illustrative example.

Roughly speaking, the algorithm computes $M$ using the equations~\eqref{eq: JHDSD 4} and~\eqref{eq: JHDSD 5}. 
The infimum in~\eqref{eq: JHDSD 4} and the supremum in~\eqref{eq: JHDSD 5} are processed by considering the elements of the respective sets in an appropriate order. 

We denote $V = \Vmin \cup \Vmax$, and for each $x \in V$ define
\begin{eqnarray*}
\Succ(x) & = & \{y \in V \mid E(x,y) > 0\} \\
\Pred(x) & = & \{y \in V \mid E(y,x) > 0\}.
\end{eqnarray*}
These sets represent the successors and predecessors of~$x$, respectively. 
For each node $x \in \Vmin$, let $x.\done$ be a Boolean flag indicating whether $x$ has been processed. 
Similarly, for each positive edge $e = \tuple{z,y} \in \Vmin \times \Vmax$, let $e.\done$ be a Boolean flag indicating whether $e$ has been processed. 
For each node $y \in \Vmax$, let $y.\counter$ denote the number of predecessors of $y$ that remain to be processed. 

The marking $M$ is initialized as follows: for each $x \in \Vmin$, set $M(x) = L(x)$; for each $y \in \Vmax$, set $M(y) = 0$. In other words, $M$ is initialized to the least favorable values from the perspective of each node. The attributes $x.\done$, $y.\counter$, and $\tuple{z,y}.\done$, for nodes $x \in \Vmin$, $y \in \Vmax$, and positive edges $\tuple{z,y} \in \Vmin \times \Vmax$, are initialized in the standard way.

\begin{algorithm}[t!]
\caption{Computing the greatest correct marking\label{alg1}}
\Input{a finite minimax net $\mN = \tuple{\Vmin, \Vmax, E, L}$ over the G\"odel structure.}
\Output{the greatest correct marking of $\mN$.}

\BlankLine

let $V = \Vmin \cup \Vmax$ and initialize $M$ to the empty fuzzy subset of $V$\label{stm: alg1 1}\;
initialize $Q$ to the empty priority queue, where smaller keys correspond to higher priority\;

\ForEach{$x \in \Vmin$}{
    $M(x) := L(x)$, $x.\done := \False$\;
    insert $x$ into $Q$ using $M(x)$ as the key\label{stm: alg1 5}\;
}
\ForEach{$y \in \Vmax$}{
    $M(y) := 0$, $y.\counter := |\Pred(y)|$\;
    \lIf{$y.\counter = 0$}{insert $y$ into $Q$ using $M(y)$ as the key}
}
\ForEach{$\tuple{z,y} \in \Vmin \times \Vmax$ such that $E(z,y) > 0$}{
    $\tuple{z,y}.\done := \False$\;
    insert $\tuple{z,y}$ into $Q$ using $E(z,y)$ as the key\label{stm: alg1 11}\;
}

\While{$Q$ is not empty}{
    extract an object (a node or an edge) $\obj$ from $Q$\;
    \uIf{$\obj$ is a node $z \in \Vmax$\label{stm: alg1 14}}{
        \ForEach{$x \in \Succ(z)$ such that $\lnot x.\done$ and $E(z,x) > M(z)$\label{stm: alg1 15}}{
            $M(x) := M(z)$\;
            insert $x$ into $Q$ using $M(x)$ as the key\label{stm: alg1 17}\;
        }
    }
    \uElseIf{$\obj$ is a node $z \in \Vmin$ such that $\lnot z.\done$\label{stm: alg1 18}}{
        \ForEach{$y \in \Succ(z)$ such that $\lnot \tuple{z,y}.\done$\label{stm: alg1 19}}{
            $M(y) := M(z)$\;
            $\tuple{z,y}.\done := \True$\;
            $y.\counter := y.\counter - 1$\;
            \lIf{$y.\counter = 0$}{insert $y$ into $Q$ using $M(y)$ as the key}
        }
        $z.\done := \True$\label{stm: alg1 24}\;
    }
    \ElseIf{$\obj$ is an edge $\tuple{z,y}$ such that $\lnot \tuple{z,y}.\done$\label{stm: alg1 25}}{
            $M(y) := E(z,y)$\label{stm: alg1 26}\;
            $\tuple{z,y}.\done := \True$\;
            $y.\counter := y.\counter - 1$\;
            \lIf{$y.\counter = 0$}{insert $y$ into $Q$ using $M(y)$ as the key\label{stm: alg1 29}}
    }
}

\Return{$M$}\;
\end{algorithm}

Nodes in $V$ and positive edges in $\Vmin \times \Vmax$ are processed using a (common) priority queue, where smaller keys correspond to higher priority. Each positive edge $\tuple{z,y} \in \Vmin \times \Vmax$ with key $E(z,y)$ is inserted into the queue initially. A node $y \in \Vmax$ with key $M(y)$ is inserted into the queue whenever $y.\counter$ becomes $0$ (or initially, if $y.\counter = 0$). Each positive edge $\tuple{z,y} \in \Vmin \times \Vmax$ and each node $y \in \Vmax$ is inserted into the queue only once and thus processed exactly once. In contrast, a node $x \in \Vmin$ with key $M(x)$ may be inserted into the queue multiple times, including its initial insertion, but it is processed only once (namely, when $x.\done = \False$). 

After initialization, as the main loop of the algorithm, while the priority queue is not empty, an object (a node or an edge) is extracted from the queue and processed as follows:
\begin{itemize}
\item With~\eqref{eq: JHDSD 4} in mind, a node $z \in \Vmax$ is processed as follows: for each successor $x \in \Vmin$ of $z$ such that $x.\done = \False$ and $E(z,x) > M(z)$, $M(x)$ is set (i.e., reduced) to $M(z)$ and $x$ is subsequently inserted into the priority queue.

\item With~\eqref{eq: JHDSD 5} in mind, a node $z \in \Vmin$ with $z.\done = \False$ is processed as follows: for each successor $y \in \Vmax$ of $z$ such that $\tuple{z,y}.\done = \False$ (which implies $E(z,y) \geq M(z)$), $M(y)$ is set (i.e., increased) to $M(z)$, the value of $y.\counter$ is decreased by~1, and the flag $\tuple{z,y}.\done$ is set to $\True$. 

\item Furthermore, again with~\eqref{eq: JHDSD 5} in mind, a positive edge $\tuple{z,y} \in \Vmin \times \Vmax$ with $\tuple{z,y}.\done = \False$ (which implies $z.\done = \False$ and $M(z) \geq E(z,y)$) is processed as follows: $M(y)$ is set (i.e., increased) to $E(z,y)$, and the value of $y.\counter$ is decreased by~1.
\end{itemize}

The algorithm is formally presented as Algorithm~\ref{alg1} (on page~\pageref{alg1}). 

The following observations are noteworthy:
\begin{itemize}
\item For $z \in \Vmax$, the value $M(z)$ is ready for use in~\eqref{eq: JHDSD 4} only when it has reached its final value, that is, after all predecessors of $z$ have been considered, i.e., when $z.\counter = 0$. Indeed, only in this case is $z$ inserted into the queue and subsequently eligible for extraction and processing.

\item Consider the processing of a node $z \in \Vmax$, and let $x \in \Vmin$ be a successor of $z$ such that $x.\done = \False$ and $E(z,x) > M(z)$. Since $x$ has not yet been processed, it appears at least once in the priority queue, and its best occurrence has key equal to the current value of $M(x)$. Hence $M(x) \geq M(z) = (E(z,x) \fto M(z))$. 
Consequently, $M(z)$ can be taken as the final value of $M(x)$. Indeed, let $z' \in \Vmax$ be another predecessor of $x$. If $z'$ has already been processed, then either $M(z') = M(z)$ or $E(z',x) \leq M(z') < M(z)$, and thus $(E(z',x) \fto M(z')) \geq (E(z,x) \fto M(z))$. If $z'$ has not yet been processed, then $M(z') \geq M(z)$, and hence $(E(z',x) \fto M(z')) \geq M(z) = (E(z,x) \fto M(z))$. 
This justifies inserting $x$ into the priority queue, as it is ready for processing. Earlier occurrences of $x$ are not removed from the queue. Whenever a node $x \in \Vmin$ is extracted from the queue, it is ignored if $x.\done = \True$.
\end{itemize}

\begin{example}\label{example: HJRKA}
Reconsider the minimax net $\mN = \tuple{\Vmin, \Vmax, E, L}$ specified in Figure~\ref{fig: HJRKA}. 
We consider the application of Algorithm~\ref{alg1} to $\mN$. 
After initialization (the statements~\ref{stm: alg1 1}--\ref{stm: alg1 11}), we have:
\begin{itemize}
\item $M = \{v_1\!:\!0.5, v_2\!:\!1, v_3\!:\!0.9, v_4\!:\!1\}$, 
\item $u.\counter = 3$ and $w.\counter = 2$, 
\item $Q$ consists of all $x \in \Vmin$ and all $\tuple{z,y} \in \Vmin \times \Vmax$ with $E(z,y) > 0$, 
\item the $\done$ attributes of all $x \in \Vmin$ and all $\tuple{z,y} \in \Vmin \times \Vmax$ with $E(z,y) > 0$ are $\False$. 
\end{itemize}
The subsequent iterations of the main ``while'' loop proceed as follows: 
\begin{itemize}
\item {\bf Iteration 1:} The min-node $v_1$, whose key (given by $M$) is $0.5$, is extracted from~$Q$ and processed. As a result, the $\done$ attributes of this min-node and its outgoing edge $\tuple{v_1,u}$ are set to $\True$, $M(u)$ is increased to $0.5$, and $u.\counter$ is reduced to~2.

\item {\bf Iteration~2:} The edge $\tuple{v_2, u}$, whose key (given by $E$) is $0.6$, is extracted from $Q$ and processed. As a result, $M(u)$ is increased to $0.6$, $\tuple{v_2, u}.\done$ is set to $\True$, and $u.\counter$ is reduced to~1.

\item {\bf Iteration~3:} The edge $\tuple{v_2, w}$, with key $0.7$, is extracted from $Q$ and processed. As a result, $M(w)$ is increased to $0.7$, $\tuple{v_2, w}.\done$ is set to $\True$, and $w.\counter$ is reduced to~1.

\item {\bf Iteration~4:} The edge $\tuple{v_4, w}$, with key $0.8$, is extracted from $Q$ and processed. As a result, $M(w)$ is increased to $0.8$, $\tuple{v_4, w}.\done$ is set to $\True$, $w.\counter$ is reduced to~0, and $w$ is inserted into~$Q$.

\item {\bf Iteration~5:} The max-node $w$, whose key (given by $M$) is $0.8$, is extracted from~$Q$ and processed. As a result, $M(v_2)$ is reduced to $0.8$, and $v_2$ is inserted into~$Q$. 

\item {\bf Iteration~6:} The min-node $v_2$, with key $0.8$, is extracted from~$Q$ and processed. As a result, $v_2.\done$ is set to $\True$, with no further changes, since $\tuple{v_2,u}.\done$ and $\tuple{v_2,w}.\done$ are already $\True$.

\item {\bf Iteration~7:} The min-node $v_3$, with key $0.9$, is extracted from~$Q$ and processed. As a result, the $\done$ attributes of this min-node and its outgoing edge $\tuple{v_3,u}$ are set to $\True$, $M(u)$ is increased to $0.9$, $u.\counter$ is reduced to~0, and $u$ is inserted into~$Q$.

\item {\bf Iteration~8:} The max-node $u$, with key $0.9$, is extracted from~$Q$ and processed. As a result, $M(v_4)$ is reduced to $0.9$, and $v_4$ is inserted into~$Q$. 

\item {\bf Iteration~9:} The min-node $v_4$, with key $0.9$, is extracted from~$Q$ and processed. As a result, $v_4.\done$ is set to $\True$, with no further changes, since $\tuple{v_4,w}.\done$ is already $\True$.

\item {\bf Iteration~10:} The edge $\tuple{v_1, u}$, with key $0.9$, is extracted from $Q$ and ignored, since its $\done$ attribute is already $\True$.\footnote{This extraction could instead occur in the iteration 7, 8, or 9.}

\item {\bf Iteration~11:} The edge $\tuple{v_3, u}$, with key $1$, is extracted from $Q$ and ignored, since its $\done$ attribute is already $\True$.

\item {\bf Iterations~12 and~13:} The second occurrences of $v_2$ and $v_4$, whose keys are equal to 1, are extracted from $Q$ and ignored, since their $\done$ attributes are already $\True$.
\end{itemize}
The main ``while'' loop then terminates, and the algorithm returns the marking 
\[ M = \{u\!:\!0.9, v_1\!:\!0.5, v_2\!:\!0.8, v_3\!:\!0.9, v_4\!:\!0.9, w\!:\!0.8\}. \]
By Theorem~\ref{theorem: JHDLZ}, stated later in this section, this is the greatest correct marking of $\mN$ under the G\"odel semantics. 
\myend
\end{example}

We have implemented Algorithm~\ref{alg1} and made the source code publicly available~\cite{CompMinimaxG-prog}. The reader may use this implementation, together with the input file ``3.in'' (also provided at~\cite{CompMinimaxG-prog}), to verify the details of the above example.\footnote{To do so, replace the line {\em ``debug = False''} at the beginning of the module {\em CompMinimaxG.py} with {\em ``debug = True''}, then run the module and examine the output corresponding to the file ``3.in''. Remember to revert this change before running the performance tests.}
Example~\ref{example: JRIKA} in the next section provides another illustration of the execution of Algorithm~\ref{alg1}.

For the priority queue $Q$ used in Algorithm~\ref{alg1}, let $\smallestKey(Q)$ denote the infimum of the keys of the elements in~$Q$. 
The following lemma captures the essence of the main loop of Algorithm~\ref{alg1}.

\begin{lemma}\label{lemma: inv}
Let $M'$ be an arbitrary correct marking of $\mN$. 
The following assertions form an invariant of the ``while'' loop of Algorithm~\ref{alg1}:
\begin{enumerate}[(a)]
\item\label{inv: a} For every $x \in \Vmin$, 
    \begin{enumerate}[(i)]
    \item\label{inv: a-i} $M(x) = L(x) \fand \bigwedge \{E(z,x) \fto M(z) \mid z \in \Vmax, z.\counter = 0, z \notin Q\}$;
    \item\label{inv: a-ii} $M(x) \geq M'(x)$;
    \item\label{inv: a-iii} if $x.\done = \False$, then $x \in Q$;
    \item\label{inv: a-iv} if $x.\done = \True$, then $M(x) \leq \smallestKey(Q)$.
    \end{enumerate}
\item\label{inv: b} For every $y \in \Vmax$, 
    \begin{enumerate}[(i)]
    \item\label{inv: b-i} $M(y) = \bigvee \{E(z,y) \fand M(z) \mid z \in \Vmin, \tuple{z,y}.\done\}$;
    \item\label{inv: b-ii} $y.\counter = \#\{z \in \Vmin \mid E(z,y) > 0, \lnot\tuple{z,y}.\done\}$;
    \item\label{inv: b-iii} if $y \in Q$, then $y.\counter = 0$.
    \end{enumerate}
\item\label{inv: c} For every $\tuple{z,y} \in \Vmin \times \Vmax$ such that $E(z,y) > 0$,
    \begin{enumerate}[(i)]
    \item\label{inv: c-i} if $\tuple{z,y}.\done = \False$, then $\tuple{z,y} \in Q$;
    \item\label{inv: c-ii} if $z.\done = \True$, then $\tuple{z,y}.\done = \True$;
    \item\label{inv: c-iii} if $\tuple{z,y}.\done = \True$ and $\tuple{z,y} \in Q$, then $z.\done = \True$;
    \item\label{inv: c-iv} if $\tuple{z,y}.\done = \True$ and $\tuple{z,y} \notin Q$, then $E(z,y) \leq \smallestKey(Q)$.
    \end{enumerate}
\end{enumerate}
\end{lemma}

\begin{proof}
It is straightforward to verify that all the assertions hold before the execution of the ``while'' loop. 
For example, the assertion~\eqref{inv: a-i} holds because $M(x) = L(x)$ and the set under the infimum operator is empty. Consequently, the assertion~\eqref{inv: a-ii} also holds. 

Now consider an arbitrary iteration of the ``while'' loop, and assume inductively that all the assertions hold prior to this iteration. We show that they continue to hold after the iteration is completed. Let $\obj$ be the object extracted from the priority queue $Q$. The following are the only cases that can occur:
\begin{itemize}
\item Case $\obj$ is a node $z \in \Vmin$ such that $z.\done = \True$, or an edge $\tuple{z,y}$ such that $\tuple{z,y}.\done = \True$: The extraction is the only modification, and it does not decrease $\smallestKey(Q)$. All assertions~\eqref{inv: a-i}--\eqref{inv: c-iv} therefore continue to hold after the iteration.

\item Case $\obj$ is a node $z \in \Vmax$: Apart from the extraction, the ``foreach'' loop~\ref{stm: alg1 15} is executed. It is clear that the assertions \eqref{inv: a-iii}, \eqref{inv: b-i}--\eqref{inv: b-iii}, and \eqref{inv: c-i}--\eqref{inv: c-iii} continue to hold after the iteration. For any $x$ satisfying the condition of the ``foreach'' loop~\ref{stm: alg1 15}, since $x.\done = \False$, by the induction assumption~\eqref{inv: a-iii}, we have $x \in Q$, and hence $M(x) \geq M(z)$. The insertion of such an $x$ to $Q$ does not decrease $\smallestKey(Q)$. Therefore, the assertions \eqref{inv: a-iv} and \eqref{inv: c-iv} continue to hold after the iteration. 

By the induction assumption \eqref{inv: b-iii}, we have $z.\counter = 0$. 
By the induction assumptions \eqref{inv: b-ii} and \eqref{inv: b-i}, it follows that $M(z) = \bigvee_{z' \in \Vmin} \left(E(z',z) \fand M(z')\right)$ before the iteration. 
Since $M'$ is a correct marking of $\mN$, we have $M'(z) \leq \bigvee_{z' \in \Vmin} \left(E(z',z) \fand M'(z')\right)$. 
By the induction assumption \eqref{inv: a-ii}, for every $z' \in \Vmin$, $M(z') \geq M'(z')$ before the iteration. Therefore, $M(z) \geq M'(z)$. 

Denote by $Z_0$ and $Z_1$ the sets $\{z' \in \Vmax \mid z'.\counter = 0, z' \notin Q\}$ before and after the extraction of $z$ from $Q$, respectively. Since $z.\counter = 0$, we have $Z_1 = Z_0 \cup \{z\}$. 

Consider the assertions \eqref{inv: a-i} and \eqref{inv: a-ii} in the case where $x.\done = \True$ or $E(z,x) \leq M(z)$. If $E(z,x) \leq M(z)$, then the value of 
\[ L(x) \fand \bigwedge \{E(z,x) \fto M(z) \mid z \in \Vmax, z.\counter = 0, z \notin Q\} \]
does not change during the iteration, and hence the assertions \eqref{inv: a-i} and \eqref{inv: a-ii} continue to hold. If $x.\done = \True$, then by the induction assumption \eqref{inv: a-iv}, we have $M(x) \leq M(z)$ before the iteration, which together with the induction assumption \eqref{inv: a-i} implies that $M(x)$ does not change during the iteration, and hence the assertions \eqref{inv: a-i} and \eqref{inv: a-ii} also continue to hold.

We now show that the assertions \eqref{inv: a-i} and \eqref{inv: a-ii} continue to hold after each iteration of the inner ``foreach'' loop~\ref{stm: alg1 15}. Consider an arbitrary iteration of this inner loop. We have $x.\done = \False$ and $E(z,x) > M(z)$. 
Before the considered iteration of the inner loop, we have:
\begin{itemize}
\item $M(x) = L(x) \fand \bigwedge \{E(z',x) \fto M(z') \mid z' \in Z_0\}$ by the induction assumption \eqref{inv: a-i}, 
\item $x \in Q$ by the induction assumption \eqref{inv: a-iii}, and hence $M(x) \geq M(z)$. 
\end{itemize}
Since $E(z,x) > M(z)$, it follows that 
\[ (E(z,x) \fto M(z)) = M(z) \leq L(x) \fand \bigwedge \{E(z',x) \fto M(z') \mid z' \in Z_0\}. \]
Since $Z_1 = Z_0 \cup \{z\}$, we further derive 
\[ M(z) = L(x) \fand \bigwedge \{E(z',x) \fto M(z') \mid z' \in Z_1\}. \]
Since $E(z,x) > M(z) \geq M'(z)$ and $M'$ is a correct marking, we have 
\[ M'(x) \leq (E(z,x) \fto M'(z)) = M'(z) \leq M(z). \] 
Therefore, after updating $M(x) := M(z)$ in the considered iteration of the inner ``foreach'' loop~\ref{stm: alg1 15}, we have 
\begin{eqnarray*}
M(x) & = & L(x) \fand \bigwedge \{E(z',x) \fto M(z') \mid z' \in Z_1\} \\
M(x) & \geq & M'(x). 
\end{eqnarray*}
This shows that the assertions \eqref{inv: a-i} and \eqref{inv: a-ii} hold after each iteration of the inner ``foreach'' loop~\ref{stm: alg1 15}, and hence also after the considered iteration of the outer ``while'' loop. 

\item Case $\obj$ is a node $z \in \Vmin$ such that $z.\done = \False$: Apart from the extraction, the ``foreach'' loop~\ref{stm: alg1 19} is executed, after which $z.\done$ is set to $\True$. Observe that this iteration does not decrease $\smallestKey(Q)$. It is clear that the assertions \eqref{inv: a-ii}--\eqref{inv: a-iv}, \eqref{inv: b-ii}, \eqref{inv: b-iii}, and \eqref{inv: c-i}--\eqref{inv: c-iii} continue to hold after the iteration. 
The assertion \eqref{inv: a-i} also continues to hold, since whenever a node $y \in \Vmax$ has $M(y)$ modified and $y.\counter$ decreased to 0, it is inserted into~$Q$. The assertion \eqref{inv: c-iv} likewise continues to hold, since before each iteration of the inner ``foreach'' loop~\ref{stm: alg1 19}, we have $\tuple{z,y}.\done = \False$, and by the induction assumption~\eqref{inv: c-i}, it follows that $\tuple{z,y} \in Q$. 

We now show that the assertion \eqref{inv: b-i} continues to hold after each iteration of the inner ``foreach'' loop~\ref{stm: alg1 19}. Consider an arbitrary iteration of this loop. Denote by $Z_0$ and $Z_1$ the sets $\{z' \in \Vmin \mid \tuple{z',y}.\done\}$ before and after this iteration. Then $Z_1 = Z_0 \cup \{z\}$. Before this iteration, we have $\tuple{z,y}.\done = \False$, and by the induction assumption \eqref{inv: c-i}, it follows that $\tuple{z,y} \in Q$, which implies $E(z,y) \geq M(z)$, and hence $E(z,y) \fand M(z) = M(z)$.    

Consider an arbitrary $z' \in Z_0$. We have $\tuple{z',y}.\done = \True$. If $\tuple{z',y} \in Q$, then by the induction assumption \eqref{inv: c-iii}, $z'.\done = \True$, and by the induction assumption \eqref{inv: a-iv}, it follows that $M(z') \leq M(z)$. If $\tuple{z',y} \notin Q$, then by the induction assumption \eqref{inv: c-iv}, $E(z',y) \leq M(z)$. In either case, $E(z',y) \fand M(z') \leq M(z)$. Therefore, 
\[ \bigvee \{E(z',y) \fand M(z') \mid z' \in Z_0\} \leq M(z). \]
Since $E(z,y) \fand M(z) = M(z)$ and $Z_1 = Z_0 \cup \{z\}$, it follows that 
\[ \bigvee \{E(z',y) \fand M(z') \mid z' \in Z_1\} = M(z). \]
Therefore, after updating $M(y) := M(z)$ and $\tuple{z,y}.\done := \True$ in the considered iteration of the inner ``foreach'' loop~\ref{stm: alg1 19}, we obtain 
\[ M(y) = \bigvee \{E(z',y) \fand M(z') \mid z' \in Z_1\}. \]
Thus, the assertion \eqref{inv: b-i} holds after each iteration of this loop, and hence also after the considered iteration of the outer ``while'' loop. 

\item Case $\obj$ is an edge $\tuple{z,y} \in \Vmin \times \Vmax$ such that $\tuple{z,y}.\done = \False$: Apart from the extraction, the statements \ref{stm: alg1 26}--\ref{stm: alg1 29} are executed. Observe that the considered iteration (of the ``while'' loop) does not decrease $\smallestKey(Q)$. It is clear that the assertions \eqref{inv: a-ii}--\eqref{inv: a-iv}, \eqref{inv: b-ii}, \eqref{inv: b-iii}, and \eqref{inv: c-i}--\eqref{inv: c-iv} continue to hold after the iteration. 
The assertion \eqref{inv: a-i} also continues to hold, since $y$ is inserted into~$Q$ when $y.\counter$ reaches 0. 

Denote by $Z_0$ and $Z_1$ the sets $\{z' \in \Vmin \mid \tuple{z',y}.\done\}$ before and after the iteration. Then $Z_1 = Z_0 \cup \{z\}$. 
Before the iteration, we have $\tuple{z,y}.\done = \False$, and by the induction assumption \eqref{inv: c-ii}, it follows that $z.\done = \False$, which by the induction assumption \eqref{inv: a-iii} implies that $z \in Q$, and consequently $M(z) \geq E(z,y)$. Thus, $E(z,y) \fand M(z) = E(z,y)$. 

Consider an arbitrary $z' \in Z_0$. We have $\tuple{z',y}.\done = \True$. If $\tuple{z',y} \in Q$, then by the induction assumption \eqref{inv: c-iii}, $z'.\done = \True$, and by the induction assumption \eqref{inv: a-iv}, it follows that $M(z') \leq E(z,y)$. If $\tuple{z',y} \notin Q$, then by the induction assumption \eqref{inv: c-iv}, $E(z',y) \leq E(z,y)$. In either case, $E(z',y) \fand M(z') \leq E(z,y)$. Therefore, 
\[ \bigvee \{E(z',y) \fand M(z') \mid z' \in Z_0\} \leq E(z,y). \]
Since $E(z,y) \fand M(z) = E(z,y)$ and $Z_1 = Z_0 \cup \{z\}$, it follows that 
\[ \bigvee \{E(z',y) \fand M(z') \mid z' \in Z_1\} = E(z,y). \]
Therefore, after updating $M(y) := E(z,y)$ and $\tuple{z,y}.\done := \True$, we obtain 
\[ M(y) = \bigvee \{E(z',y) \fand M(z') \mid z' \in Z_1\}. \]
That is, the assertion \eqref{inv: b-i} continues to hold after the considered iteration.
\myend
\end{itemize}
\end{proof}

Here we state the main theoretical result of this section:

\begin{theorem}\label{theorem: JHDLZ}
Given a finite minimax net $\mN = \tuple{\Vmin, \Vmax, E, L}$ as input, Algorithm~\ref{alg1} returns the greatest correct marking of $\mN$ under the G\"odel semantics. Moreover, assuming that the nodes in $\Vmin$ are sorted according to $L$ and that the positive edges in $\Vmin \times \Vmax$ are sorted according to $E$, Algorithm~\ref{alg1} can be implemented to run in linear time in $m$ and $n$, where $m = |E|$ and $n = |\Vmin \cup \Vmax|$. \revisedA{Without this assumption, Algorithm~\ref{alg1} can be implemented to run in time $O(m + n + l\log{l})$, where $l$ is the number of distinct fuzzy values used in the net.}
\end{theorem}

\begin{proof}
This proof relies on the assertions of Lemma~\ref{lemma: inv}. Consider the moment when Algorithm~\ref{alg1} terminates. Then $Q$ is empty. By the assertion \eqref{inv: a-iii}, we have $x.\done = \True$ for all $x \in \Vmin$. By the assertion \eqref{inv: c-ii}, it follows that $\tuple{z,y}.\done = \True$ for all $\tuple{z,y} \in \Vmin \times \Vmax$ with $E(z,y) > 0$. Consequently, by the assertion \eqref{inv: b-ii}, we have $y.\counter = 0$ for all $y \in \Vmax$. Therefore, by the assertions \eqref{inv: a-i} and \eqref{inv: b-i}, for every $x \in \Vmin$ and $y \in \Vmax$, 
\begin{eqnarray*}
M(x) & = & L(x) \fand \bigwedge \{E(z,x) \fto M(z) \mid z \in \Vmax\} \\
M(y) & = & \bigvee \{E(z,y) \fand M(z) \mid z \in \Vmin\}.
\end{eqnarray*}
Thus, the conditions~\eqref{eq: JHDSD 1}--\eqref{eq: JHDSD 3} hold, and $M$ is a correct marking of~$\mN$. 
Let $M'$ be an arbitrary correct marking of $\mN$. 
By the assertion \eqref{inv: a-ii}, for every $x \in \Vmin$ and $y \in \Vmax$, we have 
\begin{eqnarray*}
&& M(x) \geq M'(x) \\
&& M(y) = \bigvee_{z \in \Vmin}\!\!\! \left(E(z,y) \fand M(z)\right) \geq \bigvee_{z \in \Vmin}\!\!\! \left(E(z,y) \fand M'(z)\right) \geq M'(y).
\end{eqnarray*}
Therefore, $M$ is the greatest correct marking of~$\mN$.

We now analyze the complexity of Algorithm~\ref{alg1} under the assumptions stated in the theorem. 

Observe that, after the initialization (the statements~\ref{stm: alg1 1}--\ref{stm: alg1 11}), whenever an object (a node or an edge) is inserted into the priority queue $Q$, its key is less than or equal to $\smallestKey(Q)$. 
Thus, $Q$ can be implemented as a sorted list such that, after initialization, extraction is performed by removing the first element of the list, and insertion is performed at the beginning of the list. 
Consequently, the initialization can be carried out in time $O(m + n)$.

For the main ``while'' loop, observe the following:
\begin{itemize}
\item During the execution of the algorithm, for every $y \in \Vmax$, $y.\counter$ is never increased, and $y$ is inserted into $Q$ only once when $y.\counter$ reaches~0. Therefore, the total number of iterations of the inner ``foreach'' loop~\ref{stm: alg1 15} is bounded by $m$. 
\item Every positive edge $\tuple{z,y} \in \Vmin \times \Vmax$ is inserted into $Q$ only during initialization and has $\tuple{z,y}.\done$ set to $\True$ only once. Therefore, the total number of iterations of the inner ``foreach'' loop~\ref{stm: alg1 19} is also bounded by $m$. 
\item The total number of times a node in $\Vmin$ is inserted into $Q$ (by the statements~\ref{stm: alg1 5} and~\ref{stm: alg1 17}) is bounded by \mbox{$n+m$}. 
\item By the observations given above, the number of iterations of the ``while'' loop is bounded by \mbox{$2m+n$}. 
\item Each extraction from and insertion into $Q$ is performed in constant time.
\end{itemize}
Combining these observations, we conclude that the ``while'' loop runs in time \mbox{$O(m + n)$}, and so does the entire algorithm.

{\markRevisedA
Now consider the case in which the input data are not sorted. Sorting the nodes in $\Vmin$ according to $L$ and the positive edges in $\Vmin \times \Vmax$ according to $E$ takes time $O(m + n + l\log{l})$. During this sorting, the nodes and edges are grouped into buckets according to the fuzzy values assigned by $L$ and $E$, respectively. We implement the buckets as lists. It is not necessary to sort the items within each bucket; it suffices to sort the buckets according to their keys and then concatenate them. Thus, in this case, the algorithm runs in time $O(m + n + l\log{l})$.
}
\myend
\end{proof}

\section{\revisedA{An} application}
\label{section: alg2}

In~\cite{FDSML}, we studied logical characterizations and the computation of fuzzy directed simulations for fuzzy modal logics over residuated lattices. We showed that the problem of computing the greatest fuzzy directed simulation between two finite fuzzy Kripke models can be reduced to computing the greatest correct marking of the corresponding fuzzy minimax net. In this section, as an application of Algorithm~\ref{alg1} presented in the previous section, we show that it can be used to design an efficient algorithm for computing the greatest fuzzy directed simulation between two finite fuzzy labeled graphs over the G\"odel structure, and hence also for computing the greatest fuzzy directed simulation between two finite fuzzy Kripke models over the G\"odel structure.

A {\em fuzzy (labeled) graph} \cite{DBLP:journals/corr/abs-2012-01845,DBLP:journals/isci/Nguyen23} is a structure $G = \tuple{V, E, L, \SV, \SE}$, where $V$ is a non-empty set of vertices, $\SV$ (respectively, $\SE$) is a set of vertex labels (respectively, edge labels), \mbox{$E: V \times \SE \times V \to [0,1]$} is called the fuzzy set of labeled edges, and $L: V \to (\SV \to [0,1])$ is called the labeling function of vertices. 
Given vertices $x,y \in V$, a vertex label $p \in \SV$, and an edge label $r \in \SE$, $L(x)(p)$ means the degree to which $p$ is a member of the label of~$x$, and $E(x,r,y)$ means the degree to which there is an edge from $x$ to $y$ labeled by~$r$. We call a triple $\tuple{x,r,y} \in  V \times \SE \times V$ with $E(x,r,y) > 0$ a {\em positive edge}. By $|E|$ we denote the cardinality of the set of positive edges. 
For $r \in \SE$, we write $E_r$ to denote the fuzzy subset of $V \times V$ such that $E_r(x,y) = E(x,r,y)$ for $x,y \in V$. 
The graph $G$ is {\em finite} if all the sets $V$, $\SV$ and $\SE$ are finite. It is {\em image-finite} if the set $\{y \mid E_r(x,y) > 0\}$ is finite for all $r \in \SE$ and $x \in V$. 

\begin{definition}\label{def: HJHHA}
Let $G = \tuple{V, E, L, \SV, \SE}$ and $G' = \tuple{V', E', L', \SV, \SE}$ be fuzzy graphs over the same signature $\tuple{\SV, \SE}$. 
A fuzzy relation $Z : V \times V' \to [0,1]$ is called a {\em fuzzy directed simulation} between $G$ and $G'$ if the following conditions hold for all $p \in \SV$, $r \in \SE$ and all possible values of the free variables:
\begin{eqnarray}
&& Z(x,x') \leq (L(x)(p) \fto L'(x')(p)) \label{eq: FDS1} \\
&& \E y' \in V'\ (Z(x,x') \fand E_r(x,y) \leq E'_r(x',y') \fand Z(y,y')) \label{eq: FDS2} \\
&& \E y \in V\ (Z(x,x') \fand E'_r(x',y') \leq E_r(x,y) \fand Z(y,y')). \label{eq: FDS3}
\end{eqnarray}
\end{definition}

The notion defined above differs from the notion of a fuzzy bisimulation between fuzzy graphs~\cite{DBLP:journals/isci/Nguyen23,DBLP:journals/tfs/NguyenMS23} in that $\fto$ is used in~\eqref{eq: FDS1} instead of~$\fequiv$. It differs from the notion of a fuzzy simulation between fuzzy graphs~\cite{DBLP:journals/isci/Nguyen23} or fuzzy Kripke models~\cite{DBLP:journals/cas/NguyenN22} in that the condition~\eqref{eq: FDS3} is included. It serves as a counterpart of the notion of a fuzzy directed simulation between two fuzzy Kripke models introduced in~\cite{FDSML}. 

\begin{figure}[t]
\begin{center}
\begin{tabular}{|c|}
\hline
\begin{tikzpicture}[->,>=stealth,auto,black]
\node (pocz) {};
\node (G) [node distance=1.8cm, right of=pocz] {$G$};
\node (Gp) [node distance=9.0cm, right of=G] {$G'$};
\node (u) [node distance=1.5cm, below of=pocz] {$u\!:\!p_{\,0.9}$};
\node (v) [node distance=3.6cm, right of=u] {$v\!:\!q_{\,0.4}$};
\node (up) [node distance=9.0cm, right of=u] {$u'\!:\!p_{\,0.9}$};
\node (vp) [node distance=9.0cm, right of=v] {$v'\!:\!q_{\,0.5}$};
\node (wp) [node distance=3.0cm, below of=up] {$w'\!:\!q_{\,0.6}$};

\draw (u) to node [above]{\footnotesize{0.7}} (v);	
\draw (up) to node [above]{\footnotesize{0.8}} (vp);	
\draw (up) to node [left]{\footnotesize{0.6}} (wp);	
\draw (wp) edge[bend left=10] node[above, pos=0.45, xshift=-3]{\footnotesize{0.8}} (vp);
\draw (vp) edge[bend left=10] node[below, pos=0.45, xshift=3]{\footnotesize{0.7}} (wp);

\draw (u) edge[out=125,in=55,looseness=8] node[above]{\footnotesize{0.8}} (u);
\draw (v) edge[out=125,in=55,looseness=8] node[above]{\footnotesize{0.6}} (v);
\draw (up) edge[out=125,in=55,looseness=6] node[above]{\footnotesize{0.7}} (up);
\end{tikzpicture}
\\
\hline
\end{tabular}
\caption{Illustration of the fuzzy graphs used in Example~\ref{example: HGNBS}.\label{fig: HGNBS}}
\end{center}
\end{figure}
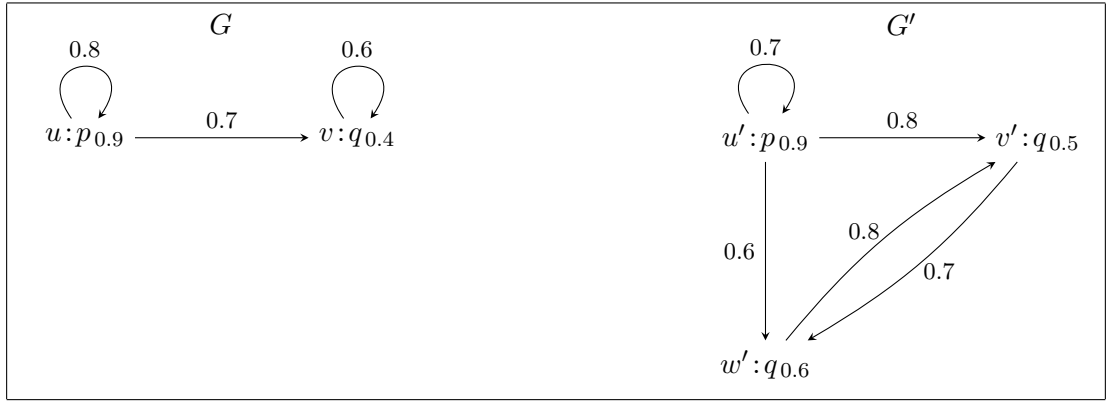

\begin{example}\label{example: HGNBS}
Consider the fuzzy graphs $G = \tuple{V, E, L, \SV, \SE}$ and $G' = \tuple{V', E', L', \SV, \SE}$ depicted in Figure~\ref{fig: HGNBS} and specified as follows:
\begin{itemize}
\item $\SV = \{p,q\}$ and $\SE = \{r\}$,  
\item $V = \{u,v\}$, $L(u) = \{p\!:\!0.9\}$, $L(v) = \{q\!:\!0.4\}$, $E = \{\tuple{u,r,u}\!:\!0.8, \tuple{u,r,v}\!:\!0.7, \tuple{v,r,v}\!:\!0.6\}$,
\item $V' = \{u',v',w'\}$, $L'(u') = \{p\!:\!0.9\}$, $L'(v') = \{q\!:\!0.5\}$, $L'(w') = \{q\!:\!0.6\}$, \\
      $E' = \{\tuple{u',r,u'}\!:\!0.7, \tuple{u',r,v'}\!:\!0.8, \tuple{u',r,w'}\!:\!0.6, \tuple{v',r,w'}\!:\!0.7, \tuple{w',r,v'}\!:\!0.8\}$.
\end{itemize}

We determine the greatest fuzzy directed simulation $Z$ between $G$ and $G'$. By~\eqref{eq: FDS1}, we have $Z(u,v') = Z(u,w') = 0$. Analogously, applying~\eqref{eq: FDS1} with $p$ replaced by $q$, we obtain $Z(v,u') = 0$. 
By~\eqref{eq: FDS3}, we have 
\begin{align*}
Z(v,v') \fand E'_r(v',w') & \leq E_r(v,v) \fand Z(v,w'), \\
Z(v,w') \fand E'_r(w',v') & \leq E_r(v,v) \fand Z(v,v'),
\end{align*}
which yields 
\begin{align*}
Z(v,v') \fand 0.7 & \leq 0.6 \fand Z(v,w'), \\
Z(v,w') \fand 0.8 & \leq 0.6 \fand Z(v,v'),
\end{align*}
and hence $Z(v,v') \leq 0.6$ and $Z(v,w') \leq 0.6$. 
Since $Z(u,v') = 0$, it follows from~\eqref{eq: FDS3} that
\[ Z(u,u') \fand E'_r(u',v') \leq E_r(u,v) \fand Z(v,v'), \]
which implies $Z(u,u') \fand 0.8 \leq 0.7 \fand 0.6$, and thus $Z(u,u') \leq 0.6$. 
Putting all these together, we obtain 
\[ Z \leq \{\tuple{u,u'}\!:\!0.6, \tuple{v,v'}\!:\!0.6, \tuple{v,w'}\!:\!0.6\}. \]
It is straightforward to verify that the fuzzy relation on the right-hand side is a fuzzy directed simulation between $G$ and $G'$. Therefore, it is the greatest fuzzy directed simulation between $G$ and $G'$ (i.e., it coincides with~$Z$). 

The reader may also use our implementation~\cite{CompMinimaxG-prog}, together with the input files ``g6.in'' and ``g6p.in'', to verify that:
\begin{itemize}
\item the greatest fuzzy simulation between $G$ and $G'$ is 
$\{\tuple{u,u'}\!:\!0.7, \tuple{v,v'}\!:\!1, \tuple{v,w'}\!:\!1\}$,
\item the greatest fuzzy bisimulation between $G$ and $G'$ is 
$\{\tuple{u,u'}\!:\!0.4, \tuple{v,v'}\!:\!0.4, \tuple{v,w'}\!:\!0.4\}$.
\myend
\end{itemize}
\end{example}

\begin{definition}\label{def: HGFHA}
Let $G = \tuple{V, E, L, \SV, \SE}$ and $G' = \tuple{V', E', L', \SV, \SE}$ be fuzzy graphs over the same signature $\tuple{\SV, \SE}$ \revisedA{such that $V \cap V' = \emptyset$}. The {\em minimax net corresponding to $\tuple{G,G'}$ via directed simulation} is $\mN = \tuple{\Vmin, \Vmax, E'', L''}$ specified as follows:
\begin{itemize}
\item $\Vmin = V \times V'$, 
\item $\Vmax = (V \times V' \times \SE) \cup (V' \times V \times \SE)$, 
\item $E'': (\Vmin \times \Vmax) \cup (\Vmax \times \Vmin) \to [0,1]$ is the following fuzzy set
\[
	\begin{array}{l}
	\{\tuple{\tuple{y,y'},\tuple{x',y,r}}\!:\!E'_r(x',y') \mid y \in V, x',y' \in V', r \in \SE\}\ \lor \\
	\{\tuple{\tuple{x',y,r},\tuple{x,x'}}\!:\!E_r(x,y) \mid x,y \in V, x' \in V', r \in \SE\}\ \lor \\
	\{\tuple{\tuple{y,y'},\tuple{x,y',r}}\!:\!E_r(x,y) \mid x,y \in V, y' \in V', r \in \SE\}\ \lor \\ 
	\{\tuple{\tuple{x,y',r},\tuple{x,x'}}\!:\!E'_r(x',y') \mid x \in V, x',y' \in V', r \in \SE\}, 
	\end{array}
\]
\item $L'': \Vmin \to [0,1]$ is the fuzzy set specified as follows, for all $x \in V$ and $x' \in V'$:
\[
	L''(\tuple{x,x'}) = \bigwedge_{p \in \SV}\!\! (L(x)(p) \fto L'(x')(p)).
\] 
\end{itemize}
\end{definition}

The notion defined above differs from the notion of a minimax net bisimulatedly corresponding to $\tuple{G,G'}$~\cite{DBLP:journals/tfs/NguyenMS23} in that $\fto$ is used in the equation defining $L''$ instead of~$\fequiv$. It serves as a counterpart of the notion of the minimax net corresponding to a pair of fuzzy Kripke models via directed simulation~\cite{FDSML}. 

\begin{remark}\label{remark: JHRJA}\markRevisedA
The assumption that $V \cap V' = \emptyset$ in the above definition is needed to distinguish the tuples in the part $(V \times V' \times \SE)$ of $\Vmax$ from the tuples in the part $(V' \times V \times \SE)$ of $\Vmax$.\footnote{\markRevisedA This is an inaccuracy in \cite[Definition~8]{DBLP:journals/tfs/NguyenMS23}, which affects \cite[Lemmas~11 and~13, Theorems~12 and~14]{DBLP:journals/tfs/NguyenMS23}. These results can be corrected by assuming that $V \cap V' = \emptyset$. Another way of correcting this is to redefine the notions of the minimax net (bi)simulatedly corresponding to $\tuple{G,G'}$ by making the changes to $\Vmax$ and $E''$ as mentioned in Remark~\ref{remark: JHRJA}.}
In the case where $V \cap V' \neq \emptyset$, one can just rename the vertices in~$V'$.
Another approach is to drop the assumption and redefine
\(\Vmax = (V \times V' \times \SE \times \{0\}) \cup (V' \times V \times \SE \times \{1\})\)
with the corresponding modifications to~$E''$.
\myend
\end{remark}

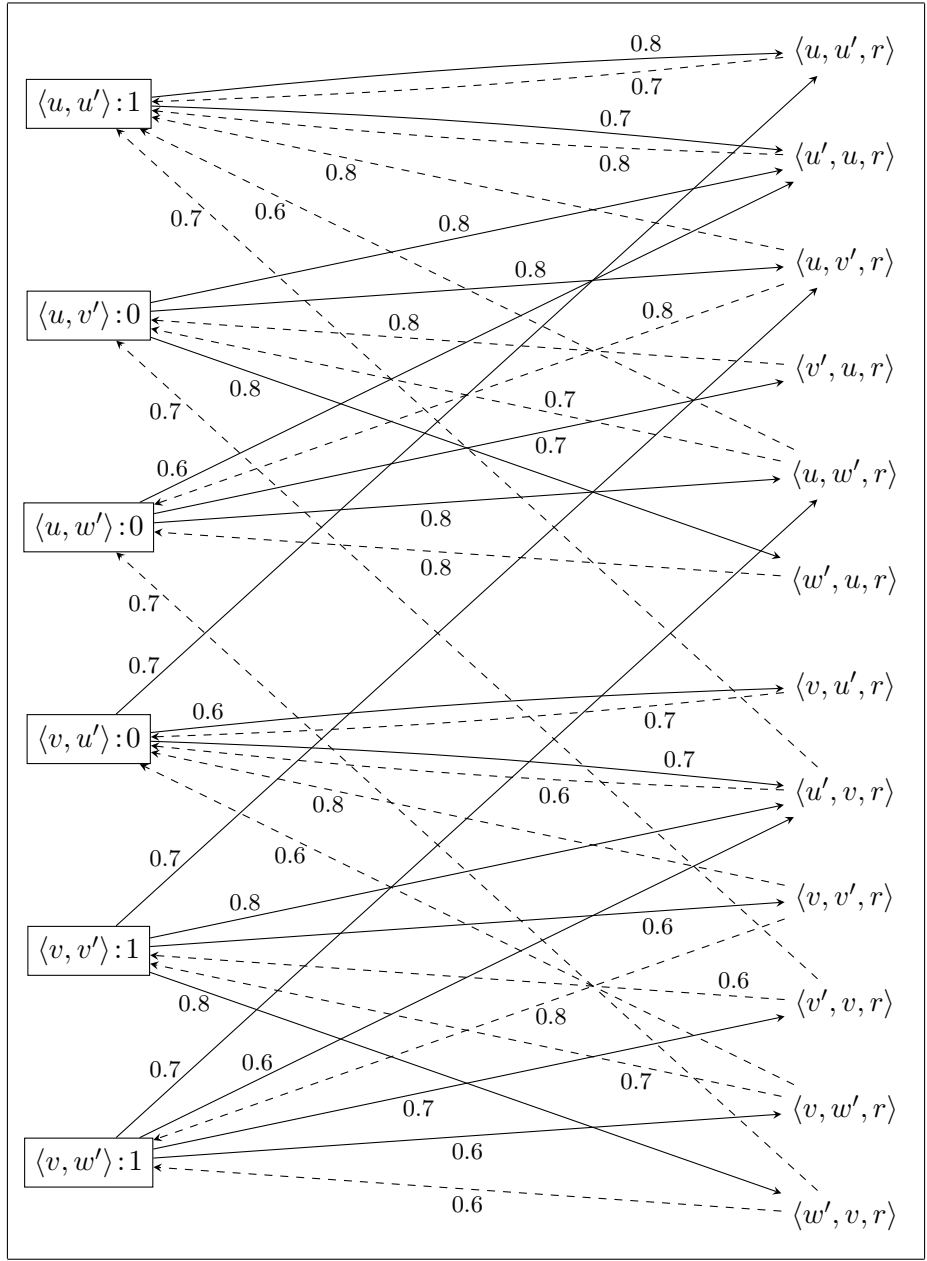
\begin{figure}[t!]
\begin{center}
\begin{tabular}{|c|}
\hline
\begin{tikzpicture}[->,>=stealth,auto,black]
\node (sta) {};
\node (uup) [draw, node distance=1.2cm, below of=sta] {$\tuple{u,u'}\!:\!1$}; 
    \node (xuup) [node distance=10cm, right of=uup] {};
    \node (uupr) [node distance=0.7cm, above of=xuup] {$\tuple{u,u',r}$};
    \node (upur) [node distance=0.7cm, below of=xuup] {$\tuple{u',u,r}$};
\node (uvp) [draw, node distance=2.8cm, below of=uup] {$\tuple{u,v'}\!:\!0$}; 
    \node (xuvp) [node distance=10cm, right of=uvp] {};
    \node (uvpr) [node distance=0.7cm, above of=xuvp] {$\tuple{u,v',r}$};
    \node (vpur) [node distance=0.7cm, below of=xuvp] {$\tuple{v',u,r}$};
\node (uwp) [draw, node distance=2.8cm, below of=uvp] {$\tuple{u,w'}\!:\!0$}; 
    \node (xuwp) [node distance=10cm, right of=uwp] {};
    \node (uwpr) [node distance=0.7cm, above of=xuwp] {$\tuple{u,w',r}$};
    \node (wpur) [node distance=0.7cm, below of=xuwp] {$\tuple{w',u,r}$};
\node (vup) [draw, node distance=2.8cm, below of=uwp] {$\tuple{v,u'}\!:\!0$}; 
    \node (xvup) [node distance=10cm, right of=vup] {};
    \node (vupr) [node distance=0.7cm, above of=xvup] {$\tuple{v,u',r}$};
    \node (upvr) [node distance=0.7cm, below of=xvup] {$\tuple{u',v,r}$};
\node (vvp) [draw, node distance=2.8cm, below of=vup] {$\tuple{v,v'}\!:\!1$}; 
    \node (xvvp) [node distance=10cm, right of=vvp] {};
    \node (vvpr) [node distance=0.7cm, above of=xvvp] {$\tuple{v,v',r}$};
    \node (vpvr) [node distance=0.7cm, below of=xvvp] {$\tuple{v',v,r}$};
\node (vwp) [draw, node distance=2.8cm, below of=vvp] {$\tuple{v,w'}\!:\!1$}; 
    \node (xvwp) [node distance=10cm, right of=vwp] {};
    \node (vwpr) [node distance=0.7cm, above of=xvwp] {$\tuple{v,w',r}$};
    \node (wpvr) [node distance=0.7cm, below of=xvwp] {$\tuple{w',v,r}$};
\node (stp) [node distance=1.0cm, below of=vwp] {};

\draw (uup) edge[bend left=2] node [above, pos=0.75, yshift=-1]{\footnotesize{0.7}} (upur);
\draw (uvp) to node [above, pos=0.48, yshift=-1]{\footnotesize{0.8}} (upur);
\draw (uvp) to node [below, pos=0.15, yshift=0]{\footnotesize{0.8}} (wpur);
\draw (uwp) to node [above, pos=0.05, yshift=0]{\footnotesize{0.6}} (upur);
\draw (uwp) to node [below, pos=0.63, yshift=1]{\footnotesize{0.7}} (vpur);

\draw (vup) edge[bend left=2] node [above, pos=0.85, yshift=-1]{\footnotesize{0.7}} (upvr);
\draw (vvp) to node [above, pos=0.15, yshift=-1]{\footnotesize{0.8}} (upvr);
\draw (vvp) to node [below, pos=0.07, yshift=1]{\footnotesize{0.8}} (wpvr);
\draw (vwp) to node [above, pos=0.18, yshift=0]{\footnotesize{0.6}} (upvr);
\draw (vwp) to node [below, pos=0.42, yshift=1]{\footnotesize{0.7}} (vpvr);

\draw (uup) edge[bend left=2] node [above, pos=0.8, yshift=-1]{\footnotesize{0.8}} (uupr);
\draw (uvp) to node [above, pos=0.6, yshift=-1]{\footnotesize{0.8}} (uvpr);
\draw (uwp) to node [below, pos=0.45, yshift=1]{\footnotesize{0.8}} (uwpr);

\draw (vup) to node [above, pos=0.04, yshift=2]{\footnotesize{0.7}} (uupr);
\draw (vvp) to node [above, pos=0.07, yshift=2]{\footnotesize{0.7}} (uvpr);
\draw (vwp) to node [above, pos=0.07, yshift=2]{\footnotesize{0.7}} (uwpr);

\draw (vup) edge[bend left=2] node [above, pos=0.08, yshift=-1]{\footnotesize{0.6}} (vupr);
\draw (vvp) to node [below, pos=0.8, yshift=1]{\footnotesize{0.6}} (vvpr);
\draw (vwp) to node [below, pos=0.5, yshift=1]{\footnotesize{0.6}} (vwpr);

\draw[dashed] (uupr) edge[bend left=2] node [below, pos=0.2, yshift=1]{\footnotesize{0.7}} (uup);	
\draw[dashed] (uvpr) to node [below, pos=0.7, yshift=1]{\footnotesize{0.8}} (uup);	
\draw[dashed] (uvpr) to node [above, pos=0.2, yshift=0]{\footnotesize{0.8}} (uwp);	
\draw[dashed] (uwpr) to node [below, pos=0.8, yshift=0]{\footnotesize{0.6}} (uup);	
\draw[dashed] (uwpr) to node [above, pos=0.35, yshift=-1]{\footnotesize{0.7}} (uvp);	

\draw[dashed] (vupr) edge[bend left=2] node [below, pos=0.18, yshift=1]{\footnotesize{0.7}} (vup);	
\draw[dashed] (vvpr) to node [below, pos=0.72, yshift=1]{\footnotesize{0.8}} (vup);	
\draw[dashed] (vvpr) to node [below, pos=0.37, yshift=1]{\footnotesize{0.8}} (vwp);	
\draw[dashed] (vwpr) to node [below, pos=0.77, yshift=0]{\footnotesize{0.6}} (vup);	
\draw[dashed] (vwpr) to node [below, pos=0.23, yshift=1]{\footnotesize{0.7}} (vvp);	

\draw[dashed] (upur) edge[bend left=2] node [below, pos=0.25, yshift=1]{\footnotesize{0.8}} (uup);	
\draw[dashed] (upvr) to node [below, pos=0.9, yshift=-3]{\footnotesize{0.7}} (uup);	
\draw[dashed] (upvr) edge[bend left=2] node [below, pos=0.35, yshift=1]{\footnotesize{0.6}} (vup);	

\draw[dashed] (vpur) to node [above, pos=0.6, yshift=-1]{\footnotesize{0.8}} (uvp);	
\draw[dashed] (vpvr) to node [below, pos=0.93, yshift=-3]{\footnotesize{0.7}} (uvp);	
\draw[dashed] (vpvr) to node [above, pos=0.08, yshift=-1]{\footnotesize{0.6}} (vvp);	

\draw[dashed] (wpur) to node [below, pos=0.55, yshift=1]{\footnotesize{0.8}} (uwp);	
\draw[dashed] (wpvr) to node [below, pos=0.96, yshift=-3]{\footnotesize{0.7}} (uwp);	
\draw[dashed] (wpvr) to node [below, pos=0.5, yshift=1]{\footnotesize{0.6}} (vwp);	
\end{tikzpicture}
\\
\hline
\end{tabular}
\caption{Illustration of the minimax net mentioned in Example~\ref{example: JAJAK}.\label{fig: JTKXH}}
\end{center}
\end{figure}

\begin{example}\label{example: JAJAK}
Reconsider the fuzzy graphs $G = \tuple{V, E, L, \SV, \SE}$ and $G' = \tuple{V', E', L', \SV, \SE}$ specified in Example~\ref{example: HGNBS} and depicted in Figure~\ref{fig: HGNBS}. The minimax net $\mN = \tuple{\Vmin, \Vmax, E'', L''}$ corresponding to $\tuple{G,G'}$ via directed simulation is illustrated in Figure~\ref{fig: JTKXH}. In this figure, min-nodes appear on the left side, whereas max-nodes are placed on the right. Each min-node is annotated with its marking limit. For clarity, only positive edges are displayed: edges in $\Vmin \times \Vmax$ are drawn with solid lines, while edges in $\Vmax \times \Vmin$ are drawn with dashed lines. The label assigned to an edge from a node $x$ to a node $y$ specifies the value $E(x,y)$. Note that the min-nodes $\tuple{u,v'}$, $\tuple{u,w'}$, and $\tuple{v,u'}$ have a marking limit of~0, so any edges incident to these nodes may be disregarded, as they have no influence on determining the greatest correct marking of~$\mN$.
\myend
\end{example}

The following lemma is a counterpart of \cite[Lemma~11]{DBLP:journals/tfs/NguyenMS23} and~\cite[Lemma~4.3]{FDSML}, which deal with fuzzy bisimulations between fuzzy graphs and fuzzy directed simulations between fuzzy Kripke models, respectively. Its proof proceeds analogously to that of \cite[Lemma~11]{DBLP:journals/tfs/NguyenMS23}. The lemma holds in a more general setting, where $\fand$ is an arbitrary t-norm. 

\begin{lemma}\label{lemma: HDJHA2}
Let $G = \tuple{V, E, L, \SV, \SE}$ and $G' = \tuple{V', E', L', \SV, \SE}$ be image-finite fuzzy graphs over the same signature \revisedA{such that $V \cap V' = \emptyset$,} and let $\mN = \tuple{\Vmin, \Vmax, E'', L''}$ be the minimax net corresponding to $\tuple{G,G'}$ via directed simulation. Let $Z$ be a fuzzy subset of $V \times V'$ and let $M: \Vmin \cup \Vmax \to [0,1]$ be
\[
\begin{array}{l}
\!\!Z \lor \{\tuple{x',y,r} \!:\! (E'_r \circ Z^{-1})(x',y) \mid x' \in V', y \in V, r \in \SE\} \\
\!\!\quad \lor\, \{\tuple{x,y',r} \!:\! (E_r \circ Z)(x,y') \mid x \in V, y' \in V', r \in \SE\}.
\end{array}
\]
Then, $Z$ is a fuzzy directed simulation between $G$ and $G'$ iff $M$ is a correct marking of~$\mN$. 
\end{lemma}	

The following theorem shows that the problem of computing the greatest fuzzy directed simulation between two finite fuzzy graphs can be reduced to the problem of computing the greatest correct marking of a finite fuzzy minimax net. 
It is a counterpart of \cite[Theorem~12]{DBLP:journals/tfs/NguyenMS23} and~\cite[Theorem~4.4]{FDSML}, which deal with fuzzy bisimulations between fuzzy graphs and fuzzy directed simulations between fuzzy Kripke models, respectively. Its proof proceeds analogously to that of \cite[Theorem~12]{DBLP:journals/tfs/NguyenMS23}, using Lemma~\ref{lemma: HDJHA2} in place of \cite[Lemma~11]{DBLP:journals/tfs/NguyenMS23}. The theorem holds in a more general setting, where $\fand$ is an arbitrary t-norm. 

\begin{theorem}\label{theorem: HDFMX 2}
Let $G = \tuple{V, E, L, \SV, \SE}$ and $G' = \tuple{V', E', L', \SV, \SE}$ be image-finite fuzzy graphs over the same signature \revisedA{such that $V \cap V' = \emptyset$,} and let $\mN = \tuple{\Vmin, \Vmax, E'', L''}$ be the minimax net corresponding to $\tuple{G,G'}$ via directed simulation. Let $M$ be the greatest correct marking of $\mN$. Then, $M|_{\Vmin}$ is the greatest fuzzy directed simulation between~$G$ and~$G'$.  
\end{theorem}

\begin{algorithm}[t]
\caption{Computing the greatest fuzzy directed simulation\label{alg2}}
\Input{finite fuzzy graphs $G = \tuple{V, E, L, \SV, \SE}$ and $G' = \tuple{V', E', L', \SV, \SE}$ \revisedA{such that $V \cap V' = \emptyset$}.}
\Output{the greatest fuzzy directed simulation between $G$ and $G'$ under the G\"odel semantics.}

\BlankLine

construct the minimax net $\mN = \tuple{\Vmin, \Vmax, E'', L''}$ corresponding to $\tuple{G,G'}$ via directed simulation\revisedA{\;}
apply Algorithm~\ref{alg1} to $\mN$ and let $M$ be the obtained result\;
\Return{$M|_{\Vmin}$}\;
\end{algorithm}

Based on the above theorem and following the method from~\cite{FDSML}, Algorithm~\ref{alg2} computes the greatest fuzzy directed simulation between two finite fuzzy graphs $G = \tuple{V, E, L, \SV, \SE}$ and $G' = \tuple{V', E', L', \SV, \SE}$ under the G\"odel semantics as follows. It first constructs the minimax net $\mN = \tuple{\Vmin, \Vmax, E'', L''}$ corresponding to $\tuple{G,G'}$ via directed simulation\revisedA{.} It then applies Algorithm~\ref{alg1} to $\mN$ to compute the greatest correct marking $M$ under the G\"odel semantics. Finally, it returns $M|_{\Vmin}$.

\begin{example}\label{example: JRIKA}
Reconsider the fuzzy graphs $G = \tuple{V, E, L, \SV, \SE}$ and $G' = \tuple{V', E', L', \SV, \SE}$ specified in Example~\ref{example: HGNBS}, 
and the minimax net $\mN = \tuple{\Vmin, \Vmax, E'', L''}$ corresponding to $\tuple{G,G'}$ via directed simulation, illustrated in Figure~\ref{fig: JTKXH} and discussed in Example~\ref{example: JAJAK}. 

In order to compute the greatest fuzzy directed simulation between $G$ and $G'$, we consider the application of Algorithm~\ref{alg1} to $\mN$. 
After initialization (the statements~\ref{stm: alg1 1}--\ref{stm: alg1 11}), we have:
\begin{itemize}
\item $M = \{\tuple{u,u'}\!:\!1, \tuple{v,v'}\!:\!1, \tuple{v,w'}\!:\!1\}$, 
\item $x.\done = \False$ for all $x \in \Vmin$, 
\item $y.\counter = |\Pred(y)|$ for all $y \in \Vmax$, 
\item $\tuple{z,y}.\done = \False$ for all $\tuple{z,y} \in \Vmin \times \Vmax$ with $E''(z,y) > 0$, 
\item $Q$ consists of all $x \in \Vmin$ and all $\tuple{z,y} \in \Vmin \times \Vmax$ with $E''(z,y) > 0$.
\end{itemize}
The subsequent iterations of the main ``while'' loop may proceed as follows (depending on the order in which elements with equal keys are extracted from the priority queue~$Q$), disregarding those iterations in which none of the blocks \ref{stm: alg1 15}--\ref{stm: alg1 17}, \ref{stm: alg1 19}--\ref{stm: alg1 24}, or \ref{stm: alg1 26}--\ref{stm: alg1 29} is executed:
\begin{itemize}
\item The min-nodes $\tuple{u,v'}$, $\tuple{u,w'}$, and $\tuple{v,u'}$, whose keys (given by $M$) are equal to~$0$, are extracted from~$Q$ and processed. The effects are as follows:
    \begin{itemize}
    \item the $\done$ attributes of these min-nodes and all their outgoing edges are set to $\True$; 
    \item $\tuple{u',v,r}.\counter$ is reduced to~2;  
    \item the $\counter$ attributes of $\tuple{u,u',r}$, $\tuple{u',u,r}$, $\tuple{u,v',r}$, and $\tuple{u,w',r}$ are reduced to~1;
    \item the $\counter$ attributes of $\tuple{v',u,r}$, $\tuple{w',u,r}$, and $\tuple{v,u',r}$ are reduced to 0, and these max-nodes are inserted into~$Q$.
    \end{itemize}

\item The max-nodes $\tuple{v',u,r}$, $\tuple{w',u,r}$, and $\tuple{v,u',r}$, whose keys (given by~$M$) are equal to~$0$, are extracted from~$Q$ and processed without any further effect.


\item The edge $\tuple{\tuple{v,v'}, \tuple{v,v',r}}$, whose key (given by $E''$) is equal to $0.6$, is extracted from $Q$ and processed. As a result, $M(\tuple{v,v',r})$ is increased to $0.6$, the $\done$ attribute of the edge is set to $\True$, $\tuple{v,v',r}.\counter$ is reduced to~0, and $\tuple{v,v',r}$ is inserted into~$Q$.

\item The edge $\tuple{\tuple{v,w'}, \tuple{u',v,r}}$, with key $0.6$, is extracted from $Q$ and processed. As a result, $M(\tuple{u',v,r})$ is increased to $0.6$, the $\done$ attribute of the edge is set to $\True$, and $\tuple{u',v,r}.\counter$ is reduced to~1.

\item The edge $\tuple{\tuple{v,w'}, \tuple{v,w',r}}$, with key $0.6$, is extracted from $Q$ and processed. As a result, $M(\tuple{v,w',r})$ is increased to $0.6$, the $\done$ attribute of the edge is set to $\True$, $\tuple{v,w',r}.\counter$ is reduced to~0, and $\tuple{v,w',r}$ is inserted into~$Q$.

\item The max-node $\tuple{v,v',r}$, with key $0.6$, is extracted from~$Q$ and processed. As a result, $M(\tuple{v,w'})$ is reduced to $0.6$, and $\tuple{v,w'}$ is inserted into~$Q$. 

\item The max-node $\tuple{v,w',r}$, with key $0.6$, is extracted from~$Q$ and processed. As a result, $M(\tuple{v,v'})$ is reduced to $0.6$, and $\tuple{v,v'}$ is inserted into~$Q$. 

\item The min-node $\tuple{v,w'}$, with key $0.6$, is extracted from~$Q$ and processed. The effects are as follows:
    \begin{itemize}
    \item the $\done$ attributes of this min-node and its outgoing edges $\tuple{\tuple{v,w'},\tuple{u,w',r}}$ and $\tuple{\tuple{v,w'},\tuple{v',v,r}}$ are set to $\True$; 
    \item $M(\tuple{u,w',r})$ and $M(\tuple{v',v,r})$ are increased to $0.6$;
    \item the $\counter$ attributes of $\tuple{u,w',r}$ and $\tuple{v',v,r}$ are reduced to~0, and these max-nodes are inserted into~$Q$.
    \end{itemize}

\item The min-node $\tuple{v,v'}$, with key $0.6$, is extracted from~$Q$ and processed. The effects are as follows:
    \begin{itemize}
    \item the $\done$ attributes of this min-node and its outgoing edges $\tuple{\tuple{v,v'},\tuple{u,v',r}}$, $\tuple{\tuple{v,v'},\tuple{u',v,r}}$, and $\tuple{\tuple{v,v'},\tuple{w',v,r}}$ are set to $\True$; 
    \item $M(\tuple{u,v',r})$ and $M(\tuple{w',v,r})$ are increased to $0.6$;
    \item the $\counter$ attributes of $\tuple{u,v',r}$, $\tuple{u',v, r}$, and $\tuple{w',v,r}$ are reduced to~0, and these max-nodes are inserted into~$Q$.
    \end{itemize}

\item The max-nodes $\tuple{u,w',r}$ and $\tuple{v',v,r}$, with key $0.6$, are extracted from~$Q$ and processed without any further effect.

\item The max-node $\tuple{u,v',r}$, with key $0.6$, is extracted from~$Q$ and processed. As a result, $M(\tuple{u,u'})$ is reduced to~$0.6$, and $\tuple{u,u'}$ is inserted into~$Q$.

\item The max-node $\tuple{u',v,r}$, with key $0.6$, is extracted from~$Q$ and processed. As a result, $\tuple{u,u'}$ is inserted into~$Q$ again.

\item The max-node $\tuple{w',v,r}$, with key $0.6$, is extracted from~$Q$ and processed without any further effect. 

\item The min-node $\tuple{u,u'}$, with key $0.6$, is extracted from~$Q$ and processed. The effects are as follows:
    \begin{itemize}
    \item the $\done$ attributes of this min-node and its outgoing edges $\tuple{\tuple{u,u'},\tuple{u,u',r}}$ and $\tuple{\tuple{u,u'},\tuple{u',u,r}}$ are set to $\True$; 
    \item $M(\tuple{u,u',r})$ and $M(\tuple{u',u,r})$ are increased to $0.6$;
    \item the $\counter$ attributes of $\tuple{u,u',r}$ and $\tuple{u',u,r}$ are reduced to~0, and these max-nodes are inserted into~$Q$.
    \end{itemize}


\item The max-nodes $\tuple{u,u',r}$ and $\tuple{u',u,r}$, with key $0.6$, are extracted from~$Q$ and processed without any further effect. 
\end{itemize}
Further iterations of the ``while'' loop do not produce any additional changes beyond extracting elements from~$Q$. The algorithm then returns 
\begin{align*}
M = \{ & \tuple{u,u'}\!:\!0.6, \tuple{v,v'}\!:\!0.6, \tuple{v,w'}\!:\!0.6, \tuple{v,v',r}\!:\!0.6, \tuple{u',v,r}\!:\!0.6, \tuple{v,w',r}\!:\!0.6, \\
& \tuple{u,w',r}\!:\!0.6, \tuple{v',v,r}\!:\!0.6, \tuple{u,v',r}\!:\!0.6, \tuple{w',v,r}\!:\!0.6, \tuple{u,u',r}\!:\!0.6, \tuple{u',u,r}\!:\!0.6 \}.
\end{align*}

Therefore, applying Algorithm~\ref{alg2} to $G$ and $G'$, we obtain 
\[ 
    M|_{\Vmin} = \{\tuple{u,u'}\!:\!0.6, \tuple{v,v'}\!:\!0.6, \tuple{v,w'}\!:\!0.6\}, 
\]
which, according to Theorem~\ref{theorem: HDFMX 2}, is the greatest fuzzy directed simulation between~$G$ and~$G'$. 
This is consistent with Example~\ref{example: HGNBS}.
\myend
\end{example}


The following theorem is a key result of this section, stating that Algorithm~\ref{alg2} correctly and efficiently computes the greatest fuzzy directed simulation between two finite fuzzy graphs under the G\"odel semantics.

\begin{theorem}\label{theorem: JHDKJ}
Algorithm~\ref{alg2} is correct. That is, given finite fuzzy graphs $G = \tuple{V, E, L, \SV, \SE}$ and $G' = \tuple{V', E', L', \SV, \SE}$ \revisedA{such that $V \cap V' = \emptyset$}, it returns the greatest fuzzy directed simulation between $G$ and $G'$ under the G\"odel semantics. \revisedA{Its} time complexity is $O((m+n)n)$, where $n = |V| + |V'|$ and $m = |E| + |E'|$, assuming that the signature $\tuple{\SV,\SE}$ is fixed.
\end{theorem}

\begin{proof}
The correctness of Algorithm~\ref{alg2} follows directly from Theorem~\ref{theorem: HDFMX 2}. 
\revisedA{Let us analyze the time complexity of Algorithm~\ref{alg2}. The minimax net $\mN = \tuple{\Vmin, \Vmax, E'', L''}$ can be constructed in time $O(mn + n^2)$, and we have $|\Vmin \cup \Vmax| = O(n^2)$ and $|E''| = O(mn)$. Let $l$ be the number of distinct fuzzy values used in $G$ and $G'$. We have $l \leq m + n\log(|\SV|)$, which is $O(m+n)$. By Theorem~\ref{theorem: JHDLZ}, the execution of Algorithm~\ref{alg1} on $\mN$ takes time $O(mn + n^2 + l\log{l})$, which is $O((m+n)n)$, since $m = O(n^2)$ and $l = O(m+n)$. The construction of $M|_{\Vmin}$ from $M$ takes time $O(n^2)$. Therefore, Algorithm~\ref{alg2} runs in time $O((m+n)n)$.}
\myend
\end{proof}


\section{Implementation and performance tests}
\label{section: impl and performance}

As mentioned earlier, we have implemented Algorithm~\ref{alg1} in Python and made the source code publicly available~\cite{CompMinimaxG-prog}. The package~\cite{CompMinimaxG-prog} includes the module {\em CompMinimaxG.py}, which defines the class {\em MinimaxNet} with the methods {\em computeGreatestMarkingG} and {\em getMarking}. The former computes the greatest correct marking of the given minimax net (represented by the {\em self} object) under the G\"odel semantics, whereas the latter extracts this marking and returns it in a simplified form. Examples illustrating the use of the class {\em MinimaxNet} are provided in the executable part of the module (i.e., inside the statement \verb|if __name__ == "__main__"|).

The package~\cite{CompMinimaxG-prog} also includes an implementation of Algorithm~\ref{alg2} previously developed for~\cite{FDSML}. In particular, the module {\em CompMinimaxG.py} also provides the function
{\small
\begin{verbatim}
    computeGreatestBisimulationOrSimulationOrDirectedSimulationG(G, G', purpose)
\end{verbatim}
}
\noindent
which computes the greatest fuzzy bisimulation, simulation, or directed simulation between two finite fuzzy graphs $G$ and $G'$ under the G\"odel semantics, according to the value of the parameter {\em purpose}, which may take one of the values ``bisimulation'', ``simulation'', and ``directed simulation''. Examples illustrating the use of this function are also provided in the executable part of the module.

Input minimax nets for Algorithm~\ref{alg1} and input fuzzy graphs for Algorithm~\ref{alg2} can be stored in text files using a simple and intuitive format. For example, the minimax net shown in Figure~\ref{fig: HJRKA} is stored in the file {\em 3.in} of~\cite{CompMinimaxG-prog}, whose contents are obtained by concatenating the columns of the following table:
\begin{center}
\small
\begin{tabular}{|l|l|l|}
\hline
\begin{tabular}{l}
v1 0.5 \\
v2 1 \\
v3 0.9 \\
v4 1 \\
\ 
\end{tabular}
& 
\begin{tabular}{l}
v1 u 0.9 \\
v2 u 0.6 \\
v3 u 1 \\
u v4 1 \\
\ 
\end{tabular}
& 
\begin{tabular}{l}
w v1 0.6 \\
w v2 0.9 \\
v2 w 0.7 \\
w v3 0.6 \\
v4 w 0.8
\end{tabular}
\\
\hline
\end{tabular}
\end{center}

Similarly, the fuzzy graph $G$ depicted in Figure~\ref{fig: HGNBS} and specified in Example~\ref{example: HGNBS} is stored in the text file {\em g6.in} of~\cite{CompMinimaxG-prog} with the following contents:
{\small
\begin{verbatim}
    u p 0.9
    v q 0.4
    
    u u r 0.8
    u v r 0.7
    v v r 0.6
\end{verbatim}
}

\begin{figure}[th!]
\centering
\includegraphics[scale=0.80]{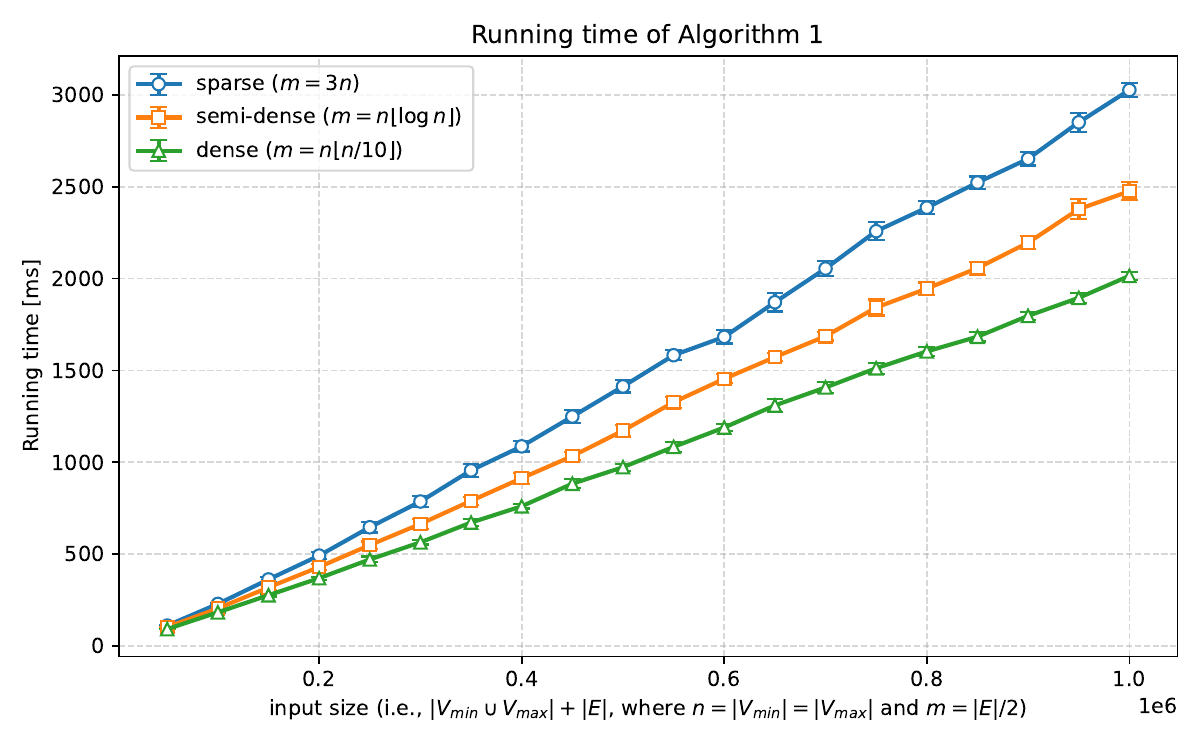}

\includegraphics[scale=0.80]{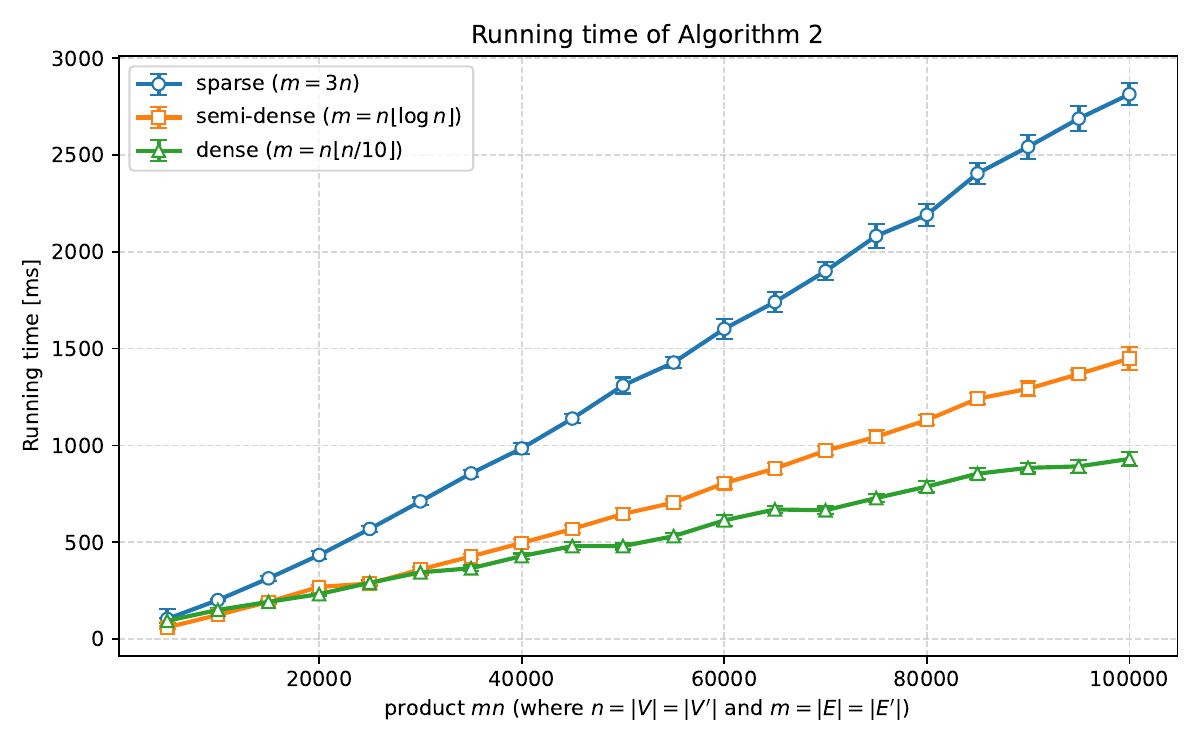}
\caption{Performance of the implementations of Algorithms~\ref{alg1} and~\ref{alg2} on random inputs\label{fig: HGFJS}}
\end{figure}

We evaluated the performance of the implementations of Algorithms~\ref{alg1} and~\ref{alg2}. All experiments were performed using a standard user account on a virtualized Linux server running Debian GNU/Linux~12 (Bookworm). The server was equipped with 64 logical Intel Broadwell CPU cores operating at approximately 2.2~GHz and 184~GiB of RAM. The user account was subject to the standard resource limits imposed by the hosting environment.

For Algorithm~\ref{alg1}, experiments were performed using randomly generated minimax nets $\mN = \tuple{\Vmin, \Vmax, E, L}$ satisfying the following properties: $|\Vmin| = |\Vmax|$, all nodes have the same outgoing degree, and all positive fuzzy values assigned in the net belong to the set $\{1/100, \ldots, 100/100\}$. Three classes of minimax nets were considered: sparse, semi-dense, and dense, for which the outgoing degree of each node in a minimax net is equal to $3$, $\lfloor\log n\rfloor$, and $\lfloor n/10\rfloor$, respectively, where $n = |\Vmin| = |\Vmax|$. Each experiment was repeated 30 times, and the average running time together with the corresponding standard deviation was measured over these repetitions. The performance results are shown in the upper part of Figure~\ref{fig: HGFJS}, where the horizontal axis represents $|\Vmin \cup \Vmax| + |E|$, i.e., the input size, and the vertical axis represents the running time in milliseconds.

For Algorithm~\ref{alg2}, experiments were performed using randomly generated fuzzy graphs $G = \tuple{V, E, L, \SV, \SE}$ and $G' = \tuple{V', E', L', \SV, \SE}$ satisfying the following properties: $|V| = |V'| = n$, $|E| = |E'| = m$, $|\SV| = |\SE| = 10$, all positive fuzzy values assigned in the graphs belong to the set $\{1/100, \ldots, 100/100\}$, and the total numbers of positive vertex labels in $G$ (i.e., $\#\{\tuple{x,p} \in V \times \SV \mid$ $L(x)(p) > 0\}$) and in $G'$ (i.e., $\#\{\tuple{x',p} \in V' \times \SV \mid$ $L'(x')(p) > 0\}$) are both equal to~$3n$. Three classes of fuzzy graphs were considered: sparse, semi-dense, and dense, where $m = 3n$, $m = n\lfloor\log n\rfloor$, and $m = n\lfloor n/10\rfloor$, respectively. Each experiment was repeated 30 times, and the average running time together with the corresponding standard deviation was measured over these repetitions. The performance results are shown in the lower part of Figure~\ref{fig: HGFJS}, where the horizontal axis represents the product $mn$, while the vertical axis represents the running time in milliseconds.

The experiments can be reproduced by executing the modules {\em experiments2.py} and {\em process$\_$results2.py} (in this order) of the package~\cite{CompMinimaxG-prog}. The former generates the files {\em results2a.txt} and {\em results2b.txt}, whereas the latter processes these files and produces the plots {\em 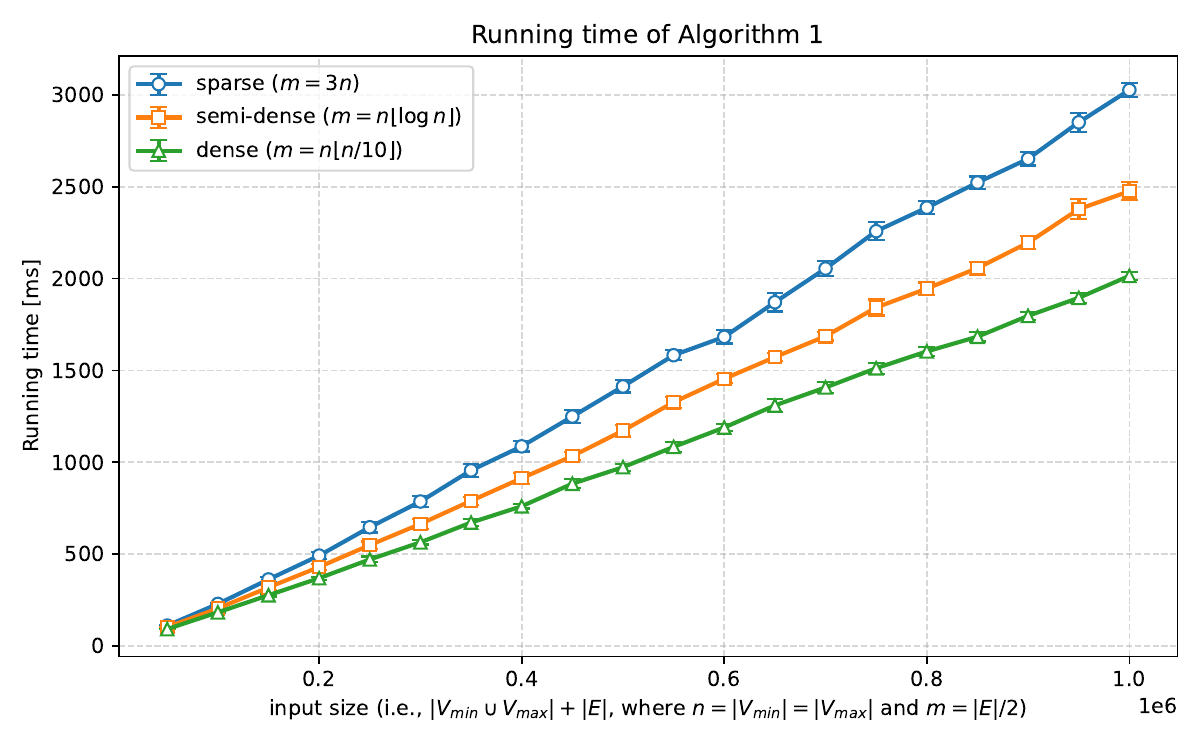} and {\em 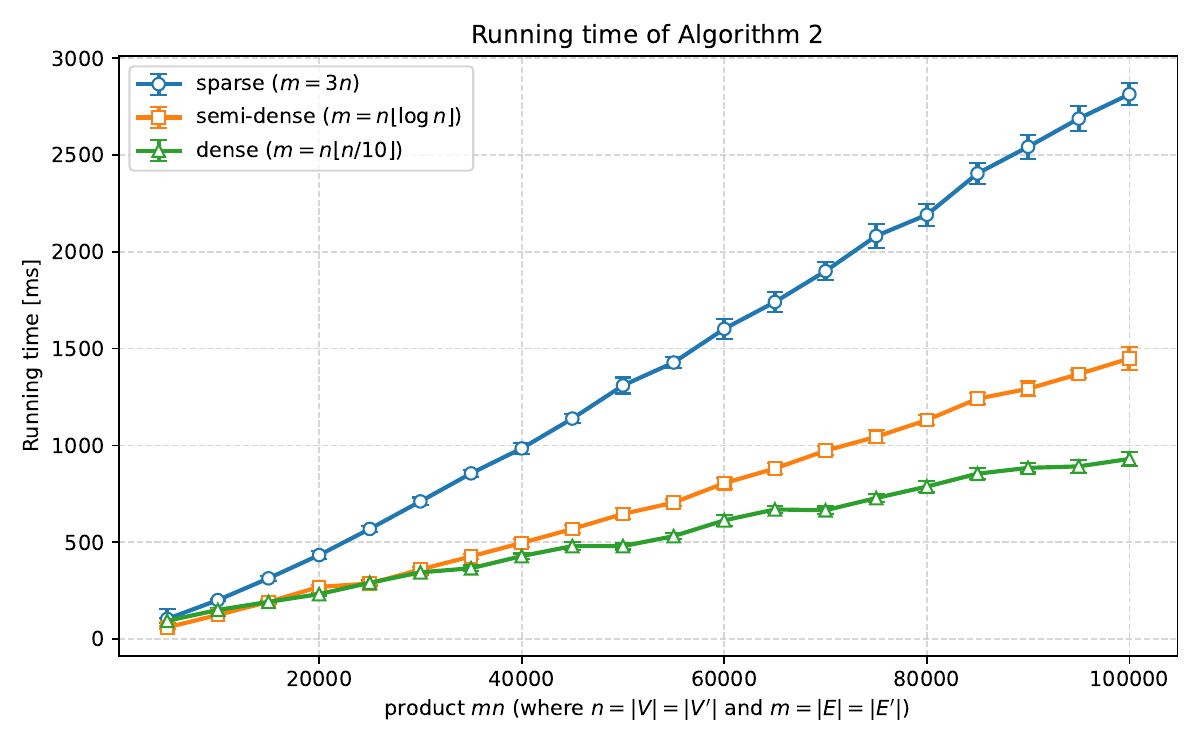}, which are used in Figure~\ref{fig: HGFJS}.

The plots in Figure~\ref{fig: HGFJS} show that the experimental results are consistent with the time complexity bounds established in Theorems~\ref{theorem: JHDLZ} and~\ref{theorem: JHDKJ}, namely, \revisedA{$O(m+n+l\log{l})$} and $O((m+n)n)$, respectively.\footnote{\markRevisedA In the experiments, we have $l = 100$.} 
\revisedA{In particular, the plot in the upper (respectively, lower) part of Figure~\ref{fig: HGFJS} gives the impression that the running time of Algorithm~\ref{alg1} (respectively, Algorithm~\ref{alg2}) appears to be linear in terms of \mbox{$m + n$} (respectively, $mn$).}

\revisedA{Let us explain the gap between the running times of Algorithm~\ref{alg2} in the sparse and dense cases. For example, when $mn \approx 10^5$, the running times in these cases were about 2814 ms and 930 ms, respectively, and the ratio of these running times is about 3.03. When $mn \approx 10^5$, in the case of sparse fuzzy graphs, we have $n = 182$ and $m = 546$, whereas in the case of dense fuzzy graphs, we have $n = 100$ and $m = 1000$. Note that, as stated in the proof of Theorem~\ref{theorem: JHDKJ}, $|\Vmin \cup \Vmax| = O(n^2)$ and $|E''| = O(mn)$. The ratio of the corresponding values of $(m+n)n$ in the two cases is about 1.20, whereas the ratio of the corresponding values of $n^2$ is about 3.31. This latter ratio heavily influences the observed performance gap.}

\section{Related work}
\label{section: related work}

The work~\cite{DBLP:journals/tfs/NguyenMS23}, which serves as the starting point for the present paper, has already been discussed in the introduction.
In this section, we review other related works on computing bisimulations, simulations, directed simulations, and related notions.

We first consider crisp graph-based structures.
\begin{itemize}
\item
Rather than computing the greatest bisimulation between two different graph-based structures, much of the literature has focused on computing the partition corresponding to the bisimilarity relation (i.e., the greatest auto-bisimulation) of a given graph-based structure.
Notable results include Hopcroft's algorithm~\cite{Hopcroft71}, which computes the state minimization of a deterministic finite (crisp) automaton in time $O(n\log n)$, and the algorithms of Paige and Tarjan~\cite{PaigeT87}, which compute the coarsest partition of a finite (crisp) graph in time $O((m+n)\log n)$ under both the stability and size-stability settings.
These settings can be referred to as the settings without counting successors and with counting successors, respectively.
The latter corresponds to the setting with graded modalities in modal logics and qualified number restrictions in description logics.
As noted in~\cite{PaigeT87}, Cardon and Crochemore~\cite{DBLP:journals/tcs/CardonC82} had previously proposed an algorithm for the latter setting with the same time complexity order.

\item
Bloom and Paige~\cite{BloomP95} and Henzinger et al.~\cite{HenzingerHK95} proposed the first algorithms with time complexity \mbox{$O((m+n)n)$} for computing the greatest auto-simulation of a finite (crisp) transition system and a finite (crisp) graph, respectively.
Both algorithms are formulated for the setting without counting successors.
Further work on computing simulations for crisp graph-based structures in this setting includes~\cite{tocl/BustanG03,jar/GentiliniPP03,iandc/RanzatoT10}, motivated, among other things, by reducing space complexity.
In~\cite{CompBSDLP}, we presented an algorithm with time complexity $O((m+n)^2n^2)$ for computing the greatest auto-simulation of a finite labeled (crisp) graph in the setting with counting successors.

\item
In~\cite{CompBSDLP}, we also presented algorithms with time complexity $O((m+n)n)$ and \mbox{$O((m+n)^2n^2)$} for computing the greatest directed auto-simulation of a finite labeled (crisp) graph in the settings without and with counting successors, respectively.
\end{itemize}

We now turn to fuzzy graph-based structures.
The computation of crisp simulations and bisimulations for fuzzy graph-based structures has been studied in
\cite{DBLP:journals/jifs/Nguyen22,DBLP:journals/ijar/NguyenT24,DBLP:journals/fss/WuCBD18,StanimirovicSC2019,NGUYEN2025109194},
whereas the computation of fuzzy simulations and bisimulations has been studied in
\cite{TFS2020,DBLP:journals/isci/Nguyen23,CiricIJD12,IgnjatovicCS15,MicicJS18,DBLP:journals/tfs/QiaoZF23,NGUYEN2025109194}.
The computation of approximate fuzzy simulations and bisimulations has been investigated in
\cite{STANKOVIC2025109299,MICIC2026109923,fss/QiaoZP23,fss/StanimirovicN0025}.
Notable algorithms with the best currently known time complexity are as follows:
\begin{itemize}
\item
In~\cite{DBLP:journals/isci/Nguyen23}, the author presented an algorithm with time complexity \mbox{$O((m\log l+n)\log n)$} for computing the fuzzy partition corresponding to the greatest fuzzy auto-bisimulation of a finite fuzzy labeled graph under the G\"odel semantics, where $l$ denotes the number of distinct edge degrees in the graph.

\item
In~\cite{DBLP:journals/ijar/NguyenT24}, the authors presented an algorithm with time complexity \mbox{$O((m\log l+n)\log n)$} for computing the crisp partition corresponding to the greatest crisp auto-bisimulation of a finite fuzzy labeled graph.

\item
In~\cite{TFS2020}, the authors presented an algorithm with time complexity \mbox{$O((m+n)n)$} for computing the greatest fuzzy simulation between two finite fuzzy interpretations in the fuzzy description logic \fALC under the G\"odel semantics.

\item
In~\cite{DBLP:journals/jifs/Nguyen22}, the author presented an algorithm with time complexity \mbox{$O((m+n)n)$} for computing the greatest crisp simulation between two finite fuzzy labeled transition systems. \revisedA{This algorithm is completely different from the above-mentioned algorithm from~\cite{TFS2020}, since the notions of crisp and fuzzy simulation are fundamentally different. The fact that their time complexity bounds are both \mbox{$O((m+n)n)$} is not surprising, since the algorithms of Bloom and Paige~\cite{BloomP95} and Henzinger et al.~\cite{HenzingerHK95} for computing the greatest auto-simulation of a finite (crisp) transition system and a finite (crisp) graph, respectively, also have this same complexity bound.}
\end{itemize}

In~\cite{DBLP:journals/corr/abs-2012-01845}, which extends~\cite{DBLP:journals/jifs/Nguyen22}, the author also presented an algorithm with time complexity \mbox{$O((m+n)n)$} for computing the greatest crisp directed simulation between two finite fuzzy labeled transition systems. 

{\markRevisedA
The algorithms given in~\cite{DBLP:journals/jifs/Nguyen22,DBLP:journals/corr/abs-2012-01845} do not lend themselves to adaptation for computing the greatest fuzzy directed simulation between two finite fuzzy graphs over the G\"odel structure, since the notions of crisp and fuzzy directed simulation are fundamentally different. In fact, Algorithms~\ref{alg1} and~\ref{alg2} are completely different from those algorithms. 

On the other hand, it is possible to modify the algorithms with time complexity $O((m + n)n)$ given in~\cite{TFS2020} for computing the greatest fuzzy (bi)simulation between two finite fuzzy interpretations in the fuzzy description logic \fALC under the G\"odel semantics to obtain an algorithm for computing the greatest fuzzy directed simulation between two finite fuzzy graphs over the G\"odel structure. However, adapting these algorithms, proving the correctness of the resulting algorithm, and analyzing its time complexity would be substantially more involved than deriving Algorithm~\ref{alg2} from Algorithm~\ref{alg1}, proving its correctness, and analyzing its time complexity.
The situation can therefore be viewed in the converse direction: due to the modularity in deriving Algorithm~\ref{alg2} from Algorithm~\ref{alg1}, Algorithm~\ref{alg1} can serve as a starting point for developing algorithms not only for computing fuzzy directed simulations but also for computing fuzzy simulations. 
In fact, the package~\cite{CompMinimaxG-prog} discussed in Section~\ref{section: impl and performance} also implements the computation of fuzzy simulations as a by-product. It also implements the computation of fuzzy bisimulations; however, its time complexity \mbox{$O((m+n)n)$} is worse than that of the algorithm in~\cite{DBLP:journals/isci/Nguyen23} for this task.
}

Finally, note that the general method for computing the greatest fuzzy directed simulation between two finite fuzzy Kripke models over any residuated lattice originates from~\cite{FDSML}, as discussed in Section~\ref{section: alg2}.
However, \revisedA{the concrete} Algorithm~\ref{alg2}, \revisedA{obtained by specializing this method to the G\"odel structure and shown to have} time complexity \mbox{$O((m+n)n)$}, is a contribution of the present paper. 
\revisedA{Algorithm~\ref{alg2} does not improve upon any of the algorithms in~\cite{DBLP:journals/jifs/Nguyen22,DBLP:journals/corr/abs-2012-01845,TFS2020} in terms of asymptotic time complexity. Rather, it is the first algorithm with this time complexity designed for computing the greatest fuzzy directed simulation between two finite fuzzy graphs over the G\"odel structure.}

\section{Conclusions}
\label{section: conc}

We have designed Algorithm~\ref{alg1} for computing the greatest correct marking of a finite fuzzy minimax net over the G\"odel structure. The algorithm runs in \revisedA{time $O(m + n + l\log{l})$, where $n$, $m$, and $l$ denote the numbers of nodes, positive edges, and distinct fuzzy values used in the net, respectively}. As \revisedA{an} application of Algorithm~\ref{alg1}, we have derived\revisedA{, by specializing the general method of \cite{FDSML} to the G\"odel structure,} the first algorithm with time complexity $O((m+n)n)$ for computing the greatest fuzzy directed simulation between two finite fuzzy graphs over the G\"odel structure, where $n$ and $m$ denote the total numbers of vertices and positive edges, respectively, in the input graphs. 

A problem closely related to the open problem raised in~\cite{DBLP:journals/tfs/NguyenMS23} is whether there exist polynomial-time algorithms for computing the greatest correct marking of a finite fuzzy minimax net over the product or \L{}ukasiewicz structure. This problem remains challenging and open, and we leave it for future work.

{\markRevisedA
\section*{Acknowledgments}

We would like to thank the anonymous reviewers for their constructive comments and suggestions, which have helped improve the quality of this paper.
}


\biboptions{sort&compress}
\bibliography{BSfDL}
\bibliographystyle{elsarticle-harv}


\end{document}